\documentclass[10pt,twocolumn,twoside]{IEEEtran}

\IEEEoverridecommandlockouts                              	

\usepackage{amssymb,amsmath,amsfonts,amsthm}
\usepackage{algorithm,algorithmic}
\usepackage{cite}
\usepackage{float}
\usepackage{epsfig}
\usepackage{graphicx}
\usepackage{graphics}
\usepackage{subfigure}
\usepackage{url}
\usepackage{xspace}
\usepackage[usenames,dvipsnames]{color}
\usepackage{soul}
\usepackage{mathtools}

\usepackage{multirow}

\usepackage{tikz}
\usetikzlibrary{arrows,quotes}

\floatstyle{ruled}
\newfloat{model}{H}{mod}
\floatname{model}{\footnotesize Model}
\newfloat{notatio}{H}{not}
\floatname{notatio}{\footnotesize Notation}

\newenvironment{varalgorithm}[1]
  {\algorithm\renewcommand{\thealgorithm}{#1}}
  {\endalgorithm}

\newenvironment{list4}{
	\begin{list}{$\bullet$}{%
			\setlength{\itemsep}{0.05cm}
			\setlength{\labelsep}{0.2cm}
			\setlength{\labelwidth}{0.3cm}
			\setlength{\parsep}{0in} 
			\setlength{\parskip}{0in}
			\setlength{\topsep}{0in} 
			\setlength{\partopsep}{0in}
			\setlength{\leftmargin}{0.16in}}}
	{\end{list}}

\newenvironment{list4a}{
	\begin{list}{$\bullet$}{%
			\setlength{\itemsep}{0.05cm}
			\setlength{\labelsep}{0.2cm}
			\setlength{\labelwidth}{0.3cm}
			\setlength{\parsep}{0in} 
			\setlength{\parskip}{0in}
			\setlength{\topsep}{0in} 
			\setlength{\partopsep}{0in}
			\setlength{\leftmargin}{0.16in}}}
	{\end{list}}

\usepackage{url,changebar,bm,xspace,dsfont}
\let\mathbb=\mathds 

\newtheorem{theorem}{Theorem}

\newtheorem{prop}{Proposition}
\newtheorem{assum}{Assumption}

\newtheorem{remark}{Remark}

\ifCLASSINFOpdf
 \else
\fi

\usepackage{amsmath,amsfonts,amssymb,amscd}
\usepackage{capt-of}
\usepackage{color}
\usepackage{cite}
\usepackage{verbatim}
\usepackage{hyperref}

\hypersetup{
    colorlinks = true,
    linkcolor = [rgb]{0,0,1},
    anchorcolor = [rgb]{0,0,1},
    citecolor = [rgb]{0.9,0.5,0},
    filecolor = [rgb]{0,0,1},
    urlcolor = [rgb]{0.9,0.5,0},
    bookmarksopen = true,
    bookmarksnumbered = true,
    breaklinks = true,
    linktocpage,
    colorlinks = true,
    linkcolor = [rgb]{0.2,0.6,0.2},
    urlcolor  = [rgb]{0,0,1},
    citecolor = [rgb]{0.9,0.5,0},
    anchorcolor = [rgb]{0.2,0.6,0.2},
}

\newtheorem{lemma}{\bfseries Lemma}

\begin{document}

\title{\LARGE \bf Distributed Optimization with Efficient Communication, Event-Triggered Solution Enhancement, and Operation Stopping}


\author{Apostolos~I.~Rikos, Wei~Jiang, Themistoklis~Charalambous, and
 Karl~H.~Johansson
    \thanks{Apostolos~I.~Rikos is with the Artificial Intelligence Thrust of the Information Hub, The Hong Kong University of Science and Technology (Guangzhou), Guangzhou, China. 
    He is also affiliated with the Department of Computer Science and Engineering, The Hong Kong University of Science and Technology, Clear Water Bay, Hong Kong. 
    E-mail: {\tt~apostolosr@hkust-gz.edu.cn}.}
    \thanks{Wei~Jiang was with the Department of Electrical Engineering and Automation, School of Electrical Engineering, Aalto University, Espoo, Finland. 
    E-mail:{\tt~wjiang.lab@gmail.com}.}
    \thanks{Themistoklis~Charalambous is with the Department of Electrical and Computer Engineering, University of Cyprus, 1678 Nicosia, Cyprus. 
    He is also a Visiting Professor with the Department of Electrical Engineering and Automation, School of Electrical Engineering, Aalto University. 
    E-mail:{\tt~charalambous.themistoklis@ucy.ac.cy}.}
    \thanks{K.~H.~Johansson is with the Division of Decision and Control Systems, KTH Royal Institute of Technology, SE-100 44 Stockholm, Sweden. 
    He is also affiliated with Digital Futures. 
    E-mail:{\tt~kallej@kth.se}.}
    \thanks{Preliminary results of this work were presented at the $2023$ IFAC World Congress \cite{2023:Rikos_Johan_IFAC}, the $2023$ IEEE Conference on Decision and Control \cite{2023:Rikos_Wei_Themis_Johan_CDC_Zooming}, and the $2024$ IEEE Conference on Decision and Control \cite{2024:Rikos_Wei_Themis_Johan_CDC_3bit}. 
    The current version of our paper includes (i) improved versions of the algorithms in \cite{2023:Rikos_Wei_Themis_Johan_CDC_Zooming} and \cite{2024:Rikos_Wei_Themis_Johan_CDC_3bit} with extended operational capabilities (i.e., the improved version of the algorithm in \cite{2023:Rikos_Wei_Themis_Johan_CDC_Zooming} is capable of terminating its operation under an updated stopping condition that also considers the gradient norms of the nodes, and in the improved version of the algorithm in \cite{2024:Rikos_Wei_Themis_Johan_CDC_3bit} the algorithm exhibits operation termination capabilities), (ii) the complete and improved versions of the convergence proofs which establish linear convergence, while allowing the step size to be chosen from a broader range of available values, (iii) an application for distributed target localization along with an extended analysis of the communication requirements of each algorithm, and (iv) extended comparisons with algorithms from the literature.}
    \thanks{This work was supported by the Guangzhou-HKUST(GZ) Joint Funding Scheme (Grant No. 2025A03J3960).
    This work was also supported by the Knut and Alice Wallenberg Foundation and the Swedish Foundation for Strategic Research. 
    It was also partially supported by the MINERVA project, funded by the European Research Council (ERC) under the Horizon 2022 research and innovation program of the European Union (Grant agreement No. 101044629).} 
}

\maketitle
\thispagestyle{empty}
\pagestyle{empty}

%
%

\begin{abstract} 
In modern large-scale systems with sensor networks and IoT devices it is essential to collaboratively solve complex problems while utilizing network resources efficiently. 
In our paper we present three distributed optimization algorithms that exhibit efficient communication among nodes. 
Our first algorithm presents a simple quantized averaged gradient procedure for distributed optimization, which is shown to converge to a neighborhood of the optimal solution. 
Our second algorithm incorporates a novel event-triggered refinement mechanism, which refines the utilized quantization level to enhance the precision of the estimated optimal solution. 
It enables nodes to terminate their operation according to predefined performance guarantees. 
Our third algorithm is tailored to operate in environments where each message consists of only a few bits. 
It incorporates a novel event-triggered mechanism for adjusting the quantizer basis and quantization level, allowing nodes to collaboratively decide operation termination based on predefined performance criteria. 
We analyze the three algorithms and establish their linear convergence. 
Finally, an application on distributed sensor fusion for target localization is used to demonstrate their favorable performance compared to existing algorithms in the literature. 
\end{abstract}

\begin{IEEEkeywords} 
Optimization, distributed algorithms, quantization, finite-time, refinement, enhancement, target localization. 
\end{IEEEkeywords}

%
%
%
%
\section{Introduction}\label{intro}



\textbf{Motivation.} 
Leveraging vast amounts of data is crucial for addressing complex challenges in large-scale systems. 
Numerous applications involve collecting data from nodes and send them to a central processor. 
The central processor performs computations and subsequently distributes the results back to the nodes 
\cite{2017_Lavaei_Sandberg_Florian, 2022_Annaswamy, 2021_Meng_Yang_Johansson_Book}. 
However, this centralized architecture is susceptible to a single point of failure, demands significant computational power, and lacks flexibility, scalability, and privacy. 
As a result, distributed data processing has emerged as a viable solution \cite{2021_Meng_Yang_Johansson_Book}. 

Distributed approaches have gained significant attention due to their ability to efficiently solve complex problems across large networked systems. 
They encompass a wide range of applications such as localization, estimation, and target tracking (see \cite{Tao:2019Survey, 2024_doostmohammadian_rikos_Johansson_survey} and references therein). 
A particularly important problem within this field is distributed optimization \cite{2018_liu_nedic_distr_opt_survey}. 
Distributed optimization enables multiple nodes to collaboratively solve optimization problems with each node only possessing partial information about the global state and objective. 
This strategy enhances scalability, performance, fault tolerance, adaptability, and efficient utilization of network resources. 
Due to its various advantages, distributed optimization has become increasingly vital across various domains, such as machine learning \cite{2020:Nedich}, control \cite{SEYBOTH:2013}, and other data-driven applications that handle substantial amounts of data \cite{2018:Stich_Jaggi}. 

Distributed optimization has received significant attention from the scientific community. 
Key approaches include primal-based methods (e.g., gradient descent) and dual-based techniques (e.g., ADMM), with this paper focusing on the primal-based methods. 
Early algorithms relied on specific network structures, such as undirected or weight-balanced directed networks \cite{2009:Nedic_Optim, 2018:Xie}. 
However, undirected communication is often impractical in wireless networks due to asymmetric signal strengths caused by device heterogeneity and interference \cite{2005:Bhaskar_book}. 
Moreover, balancing directed networks requires computationally intensive network preprocessing \cite{2014:TCNS_June, gharesifard2012distributed, 2014:ISCCSP2}.
To address these limitations, researchers have developed algorithms for directed, non-weight-balanced networks \cite{2018:Khan_AB, 2021:Nedic_PushPull}, better suited to real-world wireless communication constraints. 
These algorithms involve nodes exchanging real-valued messages and alternating between local optimization and coordination steps with neighbors. 
However, this approach can lead to slow convergence and approximate solutions near the optimum rather than exact optimality. 
For enhancing convergence speed and accuracy, recent studies introduced approaches that implement a finite-time coordination protocol between optimization steps. 
In these approaches, nodes communicate using real-valued messages through asymptotic coordination methods that converge after a specified number of iterations \cite{2024:Khatana_Salapaka_optimization, 2019:Wei_Berahas_optimization}. 
So far, a common feature of all the previously mentioned works is their reliance on exchanging real-valued messages between nodes. 
This causes high communication overhead and challenges in bandwidth- and energy-constrained environments. 
To address this, researchers developed communication-efficient distributed optimization algorithms \cite{2017:Ye_Zeilinger_Jones, 2019:Koloskova_Jaggi, 2020:Jadbabaie_Federated}. 
However, existing approaches are limited to undirected or weight-balanced directed networks.
Additionally, they primarily focus on quantizing real-valued algorithms, preserving their asymptotic convergence behavior without improving convergence rates. 
Therefore, despite reduced communication overhead, these methods often converge slowly to approximate solutions. 
Recent works aimed to address limitations in communication-efficient distributed optimization algorithms by focusing on directed, unbalanced networks. 
These approaches implement finite-time quantized coordination protocols between optimization steps \cite{2023:Bastianello_Rikos_Johansson, 2024:Bastianello_Rikos_Johansson}. 
However, they only converge to a neighborhood of the optimal solution, leaving a significant error floor that may be considerably large and hence unacceptable for the considered application. 
In contrast, the algorithms in \cite{2024_Rikos_Themis_Johan_TCNS_CPU} operate over directed unbalanced networks that compute the exact optimal solution but are restricted to quadratic local cost functions. 


To the best of our knowledge (with the exception of our works in \cite{2023:Bastianello_Rikos_Johansson, 2024:Bastianello_Rikos_Johansson, 2024_Rikos_Themis_Johan_TCNS_CPU}) the existing literature lacks distributed optimization algorithms that combine efficient communication and processing with several crucial features: (i) the ability to operate over networks without requiring undirected communication or weight-balanced properties, (ii) the capability to achieve finite-time convergence to either the exact optimal solution or a close approximation depending on the quantization level, (iii) adaptive mechanisms to refine the optimal solution estimation based on application-specific requirements, and (iv) a definitive termination criterion. 
The absence of these features in current state-of-the-art algorithms limits their applicability in realistic scenarios. 
Specifically, it results in limitations such as inefficient use of computational resources, inability to approximate the optimal solution to a level desirable for the corresponding application, and inability to transition to other tasks once the desired optimization is solved, highlighting a significant gap in the field of distributed optimization. 

\textbf{Main Contributions.} 
We aim to overcome the limitations of existing approaches in the field of distributed optimization by introducing three novel algorithms. 
Our algorithms combine gradient descent with a finite-time quantized coordination protocol. 
The main advantages of our algorithms are that they (i) exhibit efficient communication among nodes by enabling them to exchange quantized messages, and (ii) incorporate an event-triggered mechanism that enhances the estimation of the optimal solution based on predefined performance guarantees. 
Our main contributions are the following: 
\begin{list4a} 
    \item We introduce a novel distributed optimization algorithm that facilitates efficient communication among nodes by allowing them to exchange quantized messages (Algorithm~\ref{algorithm_1}). 
    We demonstrate its linear convergence rate (Theorem~\ref{theorem_convergence_stronglyConvex}). 
    We analyze its relationship between accuracy of the estimated optimal solution and the level of quantization. 
    \item Building on Algorithm~\ref{algorithm_1}, we propose an enhanced distributed optimization algorithm (Algorithm~\ref{algorithm_2}) that incorporates a novel event-triggered mechanism. 
    This mechanism allows nodes to either refine the precision of the estimated optimal solution or terminate their operation based on predefined performance guarantees. 
    Furthermore, nodes can accurately estimate the exact optimal solution with a suitable choice of parameters.
    We analyze its operation and prove its linear convergence rate (Theorem~\ref{theorem_convergence_stronglyConvex_2}). 
    \item Extending Algorithm~\ref{algorithm_1}, we develop a distributed optimization algorithm (Algorithm~\ref{algor_4}) tailored for environments with strict bandwidth constraints, where each message is limited to $N$ bits. 
    It features an adaptive event-triggered mechanism for adjusting both the quantizer basis and quantization level. 
    This enables efficient communication while ensuring distributed stopping capabilities based on performance criteria.
    \item We evaluate all proposed algorithms in a distributed target localization application, analyzing their performance in terms of convergence speed, solution precision, and communication requirements. 
    Additionally, we compare their performance against state-of-the-art algorithms from the literature.
\end{list4a} 

\textbf{Paper Organization.}
Our paper is structured as follows: 
In Section~\ref{sec:compar_operation}, we compare the operation of our proposed algorithms against existing algorithms from the literature.
In Section~\ref{sec:preliminaries}, we present the notation and essential background information. 
Section~\ref{sec:probForm} outlines the problem formulation. 
In Sections~\ref{sec:distr_opt_quant}, \ref{sec:distr_grad_zoom_quant}, and \ref{sec:distr_grad_zoom_quant_3_bit}, we introduce our distributed algorithms along with their convergence analysis. 
In Section~\ref{sec:results}, we present an application of our algorithms and compare their performance with algorithms from the literature.
Finally, we conclude the paper in Section~\ref{sec:conclusions}. 

%
%
%
%

\section{Related Works}\label{sec:compar_operation}

In this section, we compare the operation of our proposed algorithms with other algorithms from the literature. 
Subsequently, we present a table that summarizes our comparisons offering a clearer overview of the operation and highlighting the advantages of our algorithms over existing ones.


In the literature of distributed optimization, various algorithms have been proposed, each exhibiting distinct characteristics and advantages.
The study in \cite{2009:Nedic_Optim} proposes a subgradient method designed for weight-balanced directed communication networks. 
It allows nodes to exchange real-valued messages and achieves geometric convergence to the optimal solution. 
Similarly, \cite{2018:Xie} presents an asynchronous distributed gradient algorithm for undirected communication networks, enabling real-valued message exchange and attaining linear convergence. 
The work in \cite{yue2021distributed} introduces an adaptive framework for distributed optimization. 
It operates on weight-balanced directed communication networks and facilitates the exchange of real-valued messages among nodes. 
As previously mentioned, achieving undirected communication is often challenging in real-world scenarios. 
Additionally, balancing directed communication networks requires precomputing which can be computationally intensive. 
Given these challenges, researchers redirected their efforts towards developing algorithms for directed networks (not required to be weight-balanced). 
For instance, \cite{2017:Chenguang_Khan} introduces a distributed optimization algorithm for networks that are both row-stochastic and column-stochastic, enabling nodes to exchange real-valued messages. 
Similarly, \cite{nedic2017achieving} presents an algorithm designed for directed time-varying networks, supporting real-valued message exchange and achieving a linear convergence rate. 
The work in \cite{xi2017dextra} proposes an real-valued communication algorithm for directed column-stochastic networks that achieves linear convergence. 
In \cite{2018:Khan_AB}, the proposed algorithm operates over networks that are both column-stochastic and row-stochastic, achieving geometric convergence through real-valued communication among nodes. 
In a similar manner, \cite{2018:Khan_addopt} introduces an algorithm for column-stochastic networks with linear convergence. 
Finally, the work in \cite{2021:Nedic_PushPull} presents a method for column-stochastic and row-stochastic networks, enabling real-valued communication among nodes and achieving linear convergence. 
Aforementioned works follow a pattern of one local optimization step followed by one coordination step with neighboring nodes. 
This can result in slow convergence and approximate solutions near the optimum. 
To overcome this, recent studies aim to improve convergence speed and accuracy by incorporating finite-time coordination between local optimization steps. 
For instance, \cite{2019:Wei_Berahas_optimization} proposes a distributed algorithm with real-valued communication and linear convergence rate. 
However, it operates over an undirected communication network which is doubly stochastic. 
Another study, \cite{2022:Jiang_Charalambous}, operates over strongly connected directed networks, achieving linear convergence via real-valued communication. 
Lastly, \cite{2024:Khatana_Salapaka_optimization} introduces an algorithm for column-stochastic directed networks. 
It employs real-valued communication, achieving Q-linear convergence rate. 
It is important to note that, so far all existing works rely on real-valued message exchange between nodes, incurring high communication overhead. 
This poses challenges in bandwidth- and energy-constrained environments. 
To address this, researchers started developing distributed optimization algorithms that incorporate communication-efficient strategies. 
However, while the following algorithms achieve efficient communication, it often comes at the cost of requiring either undirected communication networks or directed networks that are weight-balanced. 
Specifically, the algorithm in \cite{2021:Tiancheng_Uribe} enables nodes to exchange information occasionally but operates over an undirected communication network. 
Similarly, \cite{2014:Peng_Yiguan}, \cite{2020:Jadbabaie_Federated} facilitate quantized information exchange while relying on undirected networks, with \cite{2020:Jadbabaie_Federated} requiring the existence of a parameter server. 
The work in \cite{2017:Ye_Zeilinger_Jones}  facilitates quantized information exchange over undirected networks while achieving linear convergence, aligning with the aforementioned algorithms. 
The work in \cite{2018:Huaqing_Xie} incorporates quantization and event-triggered communication but operates over weight-balanced time-varying directed graphs.
Compression-based methods are explored in \cite{2019:Koloskova_Jaggi}, which achieves linear convergence over undirected graphs, and \cite{2019:Basu_Diggavi}, which combines quantization, sparsification, and local computations but requires a master node for aggregation.
In \cite{2020:Li_Chi} the proposed efficient communication algorithm exhibits linear convergence while operating over undirected communication networks. 
Adaptive quantization is utilized in \cite{2021:Doan_Romberg} and \cite{2022:Liu_Daniel}, with the former operating over undirected graphs and the latter over weight-balanced directed networks.
Quantized communication is employed in \cite{2020:Magnusson_NaLi}, which achieves linear convergence rate while operating over an undirected graph. 
The works \cite{2020:Doostmohammadian_Charalambous} and \cite{2022:Doostmohammadian_Rikos_Johansson_Themistoklis} operate over undirected and weight-balanced directed networks, respectively, while relying on quantization for communication efficiency. 
Event-triggered approaches are explored in \cite{2022:Gao_Chaojie}, which achieves linear convergence over undirected networks. 
Consensus-based quantized optimization is proposed in \cite{2021:Kajiyama_Takai}, exhibiting linear convergence while operating on weight-balanced directed networks. 
Finally, the algorithm in \cite{2023:Ziyi_Freris} achieves communication efficiency and linear convergence while operating over undirected networks. 
However, as previously discussed, existing communication-efficient methods are limited to specific networks and often exhibit slow convergence (typically reaching approximate solutions rather than the exact optimum). 
This is due to their reliance on quantizing real-valued optimization algorithms, preserving the original asymptotic convergence properties. 
Overall, while these methods reduce communication overhead, they do not significantly improve long-term convergence rates or overall efficiency.
To address these limitations, recent works have explored approaches for directed unbalanced networks. 
These methods implement a finite-time quantized coordination algorithm between local optimization steps \cite{2023:Bastianello_Rikos_Johansson, 2024:Bastianello_Rikos_Johansson}. 
However, these works only achieve convergence to a neighborhood of the optimal solution, resulting in an error floor that may be unacceptable for certain applications.
Finally, the proposed efficient communication algorithms in \cite{2024_Rikos_Themis_Johan_TCNS_CPU} operate over directed unbalanced networks and compute the exact optimal solution.
However, they are restricted to cases where each node's local cost function is quadratic.

Our proposed Algorithm~\ref{algorithm_1}, Algorithm~\ref{algorithm_2}, and Algorithm~\ref{algor_4}, are able to operate over directed communication topologies. 
Additionally, they do not require the existence of a doubly-stochastic or weight balanced network or a master node for aggregating the computation result of every node in the network. 
This flexibility enhances their applicability across diverse scenarios. 
Finally, a notable innovation in Algorithm~\ref{algorithm_2}, and Algorithm~\ref{algor_4} are the operation termination mechanisms. 
This feature is not present in any algorithm in the literature. 
Algorithm~\ref{algorithm_2}, and Algorithm~\ref{algor_4} are the first algorithms in the literature able to terminate the optimization operation in a distributed fashion according to performance criteria that rely on convergence accuracy. 
This feature represents a significant advancement in distributed optimization algorithms, allowing for more efficient and controlled optimization processes.

In Table~\ref{table_compar} we present a summary of our comparisons. 
The focus is on: (a) the use of directed networks, (b) the absence of requirements for doubly stochastic or weight-balanced networks and a central aggregation node, (c) convergence speed, (d) communication efficiency, and (e) performance-based stopping criteria. 

\begin{table}[h] 
\centering
\caption{Comparison of our algorithms with other works in the literature. 
The metrics are: (a) utilization of directed network, (b) absence of doubly stochastic network, weight balanced network and a master node for aggregation, (c) convergence rate, (d) efficient communication, and (e) stopping capabilities according to performance.}
\label{table_compar} 
\begin{tabular}{|c|c|c|c|c|c|}
\hline 
\multirow{2}{*}{}
& (A)    & (B)    & (C)    & (D)    & (E)    \\ \hline
\cite{2009:Nedic_Optim}    & \textcolor{green}{\checkmark}    & \textcolor{red}{\text{\sffamily X}}    & Geometric   & \textcolor{red}{\text{\sffamily X}}    & \textcolor{red}{\text{\sffamily X}}    \\ 
\hline
\cite{2018:Khan_AB}    & \textcolor{green}{\checkmark}    & \textcolor{green}{\checkmark}    & Geometric   & \textcolor{red}{\text{\sffamily X}}    & \textcolor{red}{\text{\sffamily X}}    \\ 
\hline
\cite{2018:Khan_addopt}    & \textcolor{green}{\checkmark}    & \textcolor{green}{\checkmark}    & Linear   & \textcolor{red}{\text{\sffamily X}}    & \textcolor{red}{\text{\sffamily X}}    \\ 
\hline
\cite{2021:Nedic_PushPull}    & \textcolor{green}{\checkmark}    & \textcolor{green}{\checkmark}    & Linear   & \textcolor{red}{\text{\sffamily X}}    & \textcolor{red}{\text{\sffamily X}}    \\ 
\hline
\cite{2018:Xie}    & \textcolor{red}{\text{\sffamily X}}    & \textcolor{red}{\text{\sffamily X}}    & Linear   & \textcolor{red}{\text{\sffamily X}}    & \textcolor{red}{\text{\sffamily X}}    \\ 
\hline
\cite{2024:Khatana_Salapaka_optimization}    & \textcolor{green}{\checkmark}    & \textcolor{green}{\checkmark}    & Q-linear    & \textcolor{red}{\text{\sffamily X}}    & \textcolor{red}{\text{\sffamily X}}    \\ 
\hline
\cite{2022:Jiang_Charalambous}    & \textcolor{green}{\checkmark}    & \textcolor{green}{\checkmark}    & Linear    & \textcolor{red}{\text{\sffamily X}}    & \textcolor{red}{\text{\sffamily X}}    \\ 
\hline
\cite{2017:Chenguang_Khan}    & \textcolor{green}{\checkmark}    & \textcolor{green}{\checkmark}    & $\mathcal{O}\left(\frac{\ln k}{\sqrt{k}}\right)$    & \textcolor{red}{\text{\sffamily X}}    & \textcolor{red}{\text{\sffamily X}}    \\ 
\hline
\cite{2021:Tiancheng_Uribe}    & \textcolor{red}{\text{\sffamily X}}    & \textcolor{red}{\text{\sffamily X}}    & $\mathcal{O}(1 / nk)$   & \textcolor{green}{\checkmark}    & \textcolor{red}{\text{\sffamily X}}    \\ 
\hline
\cite{2014:Peng_Yiguan}    & \textcolor{red}{\text{\sffamily X}}    & \textcolor{green}{\checkmark}    & (n/a)    & \textcolor{green}{\checkmark}    & \textcolor{red}{\text{\sffamily X}}    \\ 
\hline
\cite{2017:Ye_Zeilinger_Jones}    & \textcolor{red}{\text{\sffamily X}}    & \textcolor{green}{\checkmark}    & Linear    & \textcolor{green}{\checkmark}    & \textcolor{red}{\text{\sffamily X}}    \\ 
\hline
\cite{2018:Huaqing_Xie}    & \textcolor{green}{\checkmark}    & \textcolor{red}{\text{\sffamily X}}    & $\mathcal{O}\left(\frac{\ln k}{\sqrt{k}}\right)$    & \textcolor{green}{\checkmark}    & \textcolor{red}{\text{\sffamily X}}    \\ 
\hline
\cite{2019:Koloskova_Jaggi}    & \textcolor{red}{\text{\sffamily X}}    & \textcolor{red}{\text{\sffamily X}}    & Linear    & \textcolor{green}{\checkmark}    & \textcolor{red}{\text{\sffamily X}}    \\ 
\hline
\cite{2019:Basu_Diggavi}    & \textcolor{red}{\text{\sffamily X}}    & \textcolor{red}{\text{\sffamily X}}    & $\mathcal{O}(\frac{1}{\log k})$    & \textcolor{green}{\checkmark}    & \textcolor{red}{\text{\sffamily X}}    \\ 
\hline
\cite{2020:Li_Chi}    & \textcolor{red}{\text{\sffamily X}}    & \textcolor{red}{\text{\sffamily X}}    & Linear   & \textcolor{green}{\checkmark}    & \textcolor{red}{\text{\sffamily X}}    \\ 
\hline
\cite{2021:Doan_Romberg}    & \textcolor{red}{\text{\sffamily X}}    & \textcolor{red}{\text{\sffamily X}}    & $\mathcal{O}\left(\frac{\ln k}{\sqrt{k}}\right)$    & \textcolor{green}{\checkmark}    & \textcolor{red}{\text{\sffamily X}}    \\ 
\hline
\cite{2020:Jadbabaie_Federated}    & \textcolor{red}{\text{\sffamily X}}    & \textcolor{red}{\text{\sffamily X}}    & $\mathcal{O}\left(\frac{1}{\sqrt{k}}\right)$    & \textcolor{green}{\checkmark}    & \textcolor{red}{\text{\sffamily X}}    \\ 
\hline
\cite{2020:Magnusson_NaLi}    & \textcolor{red}{\text{\sffamily X}}    & \textcolor{red}{\text{\sffamily X}}    & Linear    & \textcolor{green}{\checkmark}    & \textcolor{red}{\text{\sffamily X}}    \\ 
\hline
\cite{2022:Gao_Chaojie}    & \textcolor{red}{\text{\sffamily X}}    & \textcolor{red}{\text{\sffamily X}}    & Linear    & \textcolor{green}{\checkmark}    & \textcolor{red}{\text{\sffamily X}}    \\ 
\hline
\cite{2021:Kajiyama_Takai}    & \textcolor{green}{\checkmark}    & \textcolor{red}{\text{\sffamily X}}    & Linear    & \textcolor{green}{\checkmark}    & \textcolor{red}{\text{\sffamily X}}    \\ 
\hline
\cite{2022:Liu_Daniel}    & \textcolor{green}{\checkmark}    & \textcolor{red}{\text{\sffamily X}}    & $\mathcal{O}\left(\frac{1}{{k}^{1/2}}\right)$    & \textcolor{green}{\checkmark}    & \textcolor{red}{\text{\sffamily X}}    \\ 
\hline
\cite{2022:Doostmohammadian_Rikos_Johansson_Themistoklis}    & \textcolor{green}{\checkmark}    & \textcolor{red}{\text{\sffamily X}}    & Linear    & \textcolor{green}{\checkmark}    & \textcolor{red}{\text{\sffamily X}}    \\ 
\hline
\cite{2023:Ziyi_Freris}    & \textcolor{red}{\text{\sffamily X}}    & \textcolor{green}{\checkmark}    & Linear    & \textcolor{green}{\checkmark}    & \textcolor{red}{\text{\sffamily X}}    \\ 
\hline
Algorithm~\ref{algorithm_1}    & \textcolor{green}{\checkmark}    & \textcolor{green}{\checkmark}    & Linear    & \textcolor{green}{\checkmark}    & \textcolor{red}{\text{\sffamily X}}    \\ 
\hline
Algorithm~\ref{algorithm_2}    & \textcolor{green}{\checkmark}    & \textcolor{green}{\checkmark}    & Linear    & \textcolor{green}{\checkmark}    & \textcolor{green}{\checkmark}    \\ 
\hline
Algorithm~\ref{algor_4}    & \textcolor{green}{\checkmark}    & \textcolor{green}{\checkmark}    & Linear    & \textcolor{green}{\checkmark}    & \textcolor{green}{\checkmark}    \\ \hline 
\end{tabular} 
\end{table}

%
%
%
%
\section{NOTATION AND BACKGROUND}\label{sec:preliminaries}


\textbf{Mathematical Notation and Symbols.}  
We use the symbols $\mathbb{R}$, $\mathbb{Q}$, $\mathbb{Z}$, and $\mathbb{N}$ to represent the sets of real, rational, integer, and natural numbers, respectively. 
The set of positive integers is $\mathbb{Z}_{>0}$, the set of nonnegative real numbers is $\mathbb{R}_{\geq 0}$, and the set of rational numbers greater than one is $\mathbb{Q}_{> 1}$. 
The symbol $\mathbb{R}^n_{\geq 0}$ denotes the nonnegative orthant of the $n$-dimensional real space $\mathbb{R}^n$. 
Matrices are indicated by capital letters (e.g., $A$), and vectors are represented by lowercase letters (e.g., $x$). 
To transpose a matrix $A$ or vector $x$, we use the notation $A^\top$ and $x^\top$, respectively. 
For a matrix $A \in \mathbb{R}^{n \times n}$, the entry in row $i$ and column $j$ is represented as $a_{ij}$. 
We use $\mathbb{1}$ to refer to an all-ones vector and $\mathbb{I}$ to refer to the identity matrix of the appropriate size. 
The floor $\lfloor a \rfloor$ of a real number $a$ denotes the largest integer less than or equal to $a$, while the ceiling function $\lceil a \rceil$ denotes the smallest integer greater than or equal to $a$. 
The Euclidean norm of a vector is represented by $| \cdot |$, and the standard derivative of a function is indicated by $\nabla$. 

\textbf{Graph-Theoretic Notions.} 
The communication network is described by a directed graph (digraph) denoted as $\mathcal{G} = (\mathcal{V}, \mathcal{E})$. 
The set of nodes is represented as $\mathcal{V} = \{ v_1, v_2, ..., v_n \}$, and the set of edges is defined as $\mathcal{E} \subseteq \mathcal{V} \times \mathcal{V} \cup \{ (v_i, v_i) \ | \ v_i \in \mathcal{V} \}$ (each node has a virtual self-edge). 
The number of nodes and edges is $| \mathcal{V} | = n$ and $| \mathcal{E} | = m$, respectively. 
A directed edge from node $v_i$ to node $v_l$ is represented by $(v_l, v_i) \in \mathcal{E}$.
It signifies that node $v_l$ can receive information from node $v_i$ at time step $k$, but not the reverse. 
The set of nodes that can directly transmit information to node $v_i$ is the in-neighbors of $v_i$, denoted as $\mathcal{N}_i^- = \{ v_j \in \mathcal{V} ; | ; (v_i, v_j)\in \mathcal{E}\}$. 
Similarly, the set of nodes that can directly receive information from $v_i$ is the out-neighbors of $v_i$, represented by $\mathcal{N}_i^+ = \{ v_l \in \mathcal{V} ; | ; (v_l, v_i)\in \mathcal{E}\}$. 
The in-degree and out-degree of $v_i$ are denoted $\mathcal{D}_i^- = | \mathcal{N}_i^- |$ and $\mathcal{D}_i^+ = | \mathcal{N}_i^+ |$, respectively.
A directed path from $v_i$ to $v_l$ with a length of $t$ exists if a sequence of nodes $i \equiv l_0,l_1, \dots, l_t \equiv l$ can be found, satisfying $(l_{\tau+1},l_{\tau}) \in \mathcal{E}$ for $ \tau = 0, 1, \dots , t-1$. 
The diameter $D$ of $\mathcal{G}$ is defined as the longest shortest path between any two nodes $v_i$ and $v_l$ within $\mathcal{V}$.
A digraph is strongly connected if there exists a directed path from every node $v_i$ to every other node $v_l$, for all pairs of nodes $v_i, v_l \in \mathcal{V}$. 

\textbf{Quantizers.} 
Quantization serves as a technique for reducing the number of bits necessary to represent information. 
It can enhance transmission power efficiency and computational effectiveness. 
%
There are many types of quantizers and can be categorized based on different characteristics, such as the uniformity of the step/interval size, the symmetry with respect to its basis, and the adaptability. 
Even though the results of this paper work for different types of quantizers (e.g., logarithmic), in this paper, we will consider symmetric uniform\footnote{Uniform quantizers partition the value range into equally-sized intervals.} scalar quantizers. 
Such quantizers can be further divided into mid-tread\footnote{A type of uniform quantizer where small input values (close to zero) are quantized to zero.} and uniform mid-rise\footnote{A type of uniform quantizer where the quantization levels do not include zero, meaning that small input values are not mapped to zero.} quantizers \cite{jayant1984digital}. 
For this paper, we will focus on symmetric mid-rise quantizers \cite{2019:Wei_Johansson} with \emph{i)} infinite range and \emph{ii)} bounded range ($N$-bits) and base shifting capabilities \cite{jayant1984digital}. 

First, we consider the quantization with infinite range to illustrate how our algorithms operate with quantized communication. 
The symmetric mid-rise quantizer with infinite range can be represented as  
\begin{equation}\label{asymm_quant_defn} 
    q_{\Delta}^s(\xi) = \Delta \Bigl ( \Bigl \lfloor \frac{\xi}{\Delta} \Bigr \rfloor + \frac{1}{2} \Bigr ) , 
\end{equation}
where $\xi \in \mathbb{R}$ represents the value to be quantized, the quantization level, and $q_{\Delta}^s(\xi) \in \mathcal{Q}$ is a countably infinite set and can be represented by bits through a coding scheme $b:\mathcal{Q}\mapsto \mathcal{B}$, where $\mathcal{B}$ the set of binary representations of the elements in $\mathcal{Q}$. 

We now consider a quantizer with a bounded range, which is particularly useful in practical applications where bandwidth constraints are a concern. For simplicity, we consider fixed-length coding (FLC), i.e., the same number of bits are assigned to every quantization step, regardless of how frequently it appears (e.g., if we have $8$ quantization steps, each would be assigned $3$ bits). 
The $N$-bit mid-rise uniform quantizer with a bounded range is defined as
\begin{align*} 
Q^{\text{$N$MRU}}&(b_q, \xi, \Delta) =b_q \\ 
&+ \left( \max\left(-N, \min\left(N, \left\lfloor \frac{\xi}{\Delta} \right\rfloor\right)\right) + \frac{1}{2} \right) \Delta,
\end{align*} 
where $\xi \in \mathbb{R}$ is the value to be quantized, $b_q \in \mathbb{Q}$ the basis of the quantizer,  $Q^{\text{$N$MRU}}(b_q, \xi,\Delta)$ the quantized version of $\xi$, and $\Delta \in \mathbb{Q}$ the quantization level. 
For example, the $3$-bit mid-rise uniform quantizer 
is defined as 
\begin{align*}
&Q^{\text{$3$MRU}}(b_q, \xi, \Delta)  \\
    &=\left\{ \begin{array}{ll}
         b_q - 7 \Delta / 2, & \xi \in (-\infty, b_q-3\Delta) \\ \ \vspace{-.3cm} \\ 
         b_q - 5 \Delta / 2, & \xi \in [b_q-3\Delta, b_q-2\Delta) \\ \ \vspace{-.3cm} \\ 
         b_q - 3 \Delta / 2, & \xi \in [b_q-2\Delta, b_q-\Delta) \\ \ \vspace{-.3cm} \\ 
         b_q - \Delta / 2, & \xi \in [b_q-\Delta, b_q) \\ \ \vspace{-.3cm} \\ 
         b_q + \Delta / 2, & \xi \in [b_q, b_q + \Delta) \\ \ \vspace{-.3cm} \\ 
         b_q + 3 \Delta / 2, & \xi \in [b_q + \Delta, b_q + 2\Delta) \\ \ \vspace{-.3cm} \\ 
         b_q + 5 \Delta / 2, & \xi \in [b_q + 2 \Delta, b_q + 3 \Delta) \\ \ \vspace{-.3cm} \\ 
         b_q + 7 \Delta / 2, & \xi \in [b_q + 3 \Delta, + \infty) \\ \ \vspace{-.3cm} \\
    \end{array} \right.
\end{align*} 
where $\xi \in \mathbb{R}$ is the value to be quantized, $b_q \in \mathbb{Q}$ the basis of the quantizer,  $Q^{\text{$3$MRU}}(b_q, \xi,\Delta)$ the quantized version of $\xi$, and $\Delta \in \mathbb{Q}$ the quantization level. 
The binary representation using the FLC scheme is denoted by $b^{\text{FLC}}$, such that $b^{\text{FLC}}:\mathcal{Q}\mapsto \mathcal{B}^{3}$, where $\mathcal{B}^{3}$ denotes the set of all $3$-bit binary values. 
A graphical representation is shown in Fig.~\ref{3bit_Mid_Rise_Uniform_Quantizer}. 

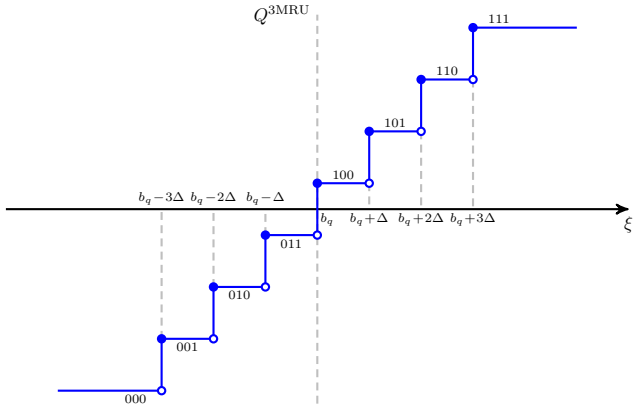
\begin{figure}[h]
\centering
\begin{tikzpicture}[scale=0.69,transform shape,
    thick,
    >=stealth',
     empty dot/.style = { circle, draw, fill = white!0,
                          inner sep = 0pt, minimum size = 4pt },
    filled dot/.style = { empty dot, fill = black}
  ]
  \def\r{2}
  \draw[->] (-6,0) -- (6,0) coordinate[label = {below:$\xi$}] (xmax);
  \draw[densely dashed, draw=lightgray] (0,-3.75) -- (0,3.75) coordinate[label = {left:$Q^{3\mathrm{MRU}}$}]  (ymax);
 \draw [draw=blue] (3,3.5) -- (3,2.5);
\node [filled dot, draw=blue, fill=blue] at (3,3.5) {};
\draw  [draw=blue] (3,3.5) -- (5,3.5); 
  \foreach \i in {\r+1,\r,\r-1} {
    \draw [draw=blue] (\i-1,\i-1/2)   -- (\i,\i-1/2);
    \draw [draw=blue] (\i-1,\i-1/2)   -- (\i-1,\i-3/2);
    \draw [densely dashed, draw=lightgray] (\i,\i-1/2) -- (\i,0);
    \node [filled dot, draw=blue, fill=blue] at (\i-1,\i-1/2) {};
    \node [empty  dot,draw=blue] at (\i,\i-1/2) {};
  }
   \draw[draw=blue] (-3,-3.5) -- (-5,-3.5);
  
  \foreach \i in {\r-5,\r-4,\r-3} {
    \draw [draw=blue] (\i,\i+1/2)   -- (\i+1,\i+1/2);
    \draw [draw=blue] (\i,\i+1/2)  -- (\i,\i-1/2);
    \draw [densely dashed, draw=lightgray] (\i,\i+1/2) -- (\i,0);
    \node [filled dot, draw=blue, fill=blue] at (\i,\i+1/2) {};
  }
  \foreach \i in {\r-6,\r-5,\r-4,\r-3} {
  \node [empty  dot,draw=blue] at (\i+1,\i+1/2) {};
 }

  \node ["right:\scriptsize $b_q$"] at (\r-2.2,-0.2) {};
  \node ["below:\scriptsize $b_q\!+\!\Delta$"]     at (\r-1,0.18)   {};
  \node ["below:\scriptsize $b_q\!+\!2\Delta$"]     at (\r,0.18)   {};
  \node ["below:\scriptsize $b_q\!+\!3\Delta$"]     at (\r+1,0.18)   {};
    
  \node ["above:\scriptsize $b_q\!-\!\Delta$"]     at (\r-3,-0.18)   {};  
  \node ["above:\scriptsize $b_q\!-\!2\Delta$"]     at (\r-4,-0.18)   {}; 
  \node ["above:\scriptsize $b_q\!-\!3\Delta$"]     at (\r-5,-0.18)   {};  
  
  \node ["below:\scriptsize $000$"]     at (\r-5.5,-3.32)   {};  
  \node ["below:\scriptsize $001$"]     at (\r-4.5,-2.32)   {};  
  \node ["below:\scriptsize $010$"]     at (\r-3.5,-1.32)   {};  
   \node ["below:\scriptsize $011$"]     at (\r-2.5,-0.32)   {}; 
   
   \node ["above:\scriptsize $100$"]     at (\r-1.5,0.32)   {}; 
   \node ["above:\scriptsize $101$"]     at (\r-0.5,1.32)   {}; 
   \node ["above:\scriptsize $110$"]     at (\r+0.5,2.32)   {}; 
    \node ["above:\scriptsize $111$"]     at (\r+1.5,3.32)   {}; 
\end{tikzpicture}
\caption{A $3$-bit mid-rise uniform quantizer, $Q^{\text{$3$MRU}}(b_q, \xi,\Delta)$. 
The fixed-length coding scheme (FLC) is also shown for each quantization step.}  
\label{3bit_Mid_Rise_Uniform_Quantizer}
\end{figure}

An important feature of $Q^{\text{$N$MRU}}(b_q, \xi, \Delta)$ compared to $q_{\Delta}^s(\xi)$ is its base-shifting capabilities (i.e., we are able to adjust the value of $b_q$). 
As we will see later in the paper (see Section~\ref{sec:distr_grad_zoom_quant_3_bit}), this characteristic allows us to introduce adaptive quantization strategies with respect to both the quantization step and basis for designing communication efficient algorithms.

%
%
%
%
\section{Problem Formulation}\label{sec:probForm}

We now introduce the optimization problems that serve as the foundation for our proposed distributed algorithms.
Specifically, we outline three problems (named \textbf{P1}, \textbf{P2}, \textbf{P3} below) that address different aspects of distributed optimization under quantized communication.

\subsection{Optimization Problems}

Let us consider a network represented as a digraph $\mathcal{G} = (\mathcal{V}, \mathcal{E})$ and the following setting: 
\begin{list4} 
\item Each node $v_i$ possesses a local cost function $f_i(x): \mathbb{R}^p \mapsto \mathbb{R}$, which is exclusively known to that specific node.   
\item Communication channels among nodes are constrained in capacity, such that only quantized values that can be expressed as rational numbers are allowed for transmission among nodes. 
\end{list4}




\noindent
\textbf{Optimization Problem~P1.}
Our first objective is to design a distributed algorithm that enables nodes to collaboratively find an approximate solution to the following optimization problem:
\begin{subequations}
\begin{align}
\mathbf{P1:}~\min_{x\in \mathcal{X}}~ & F(x_1, x_2, ..., x_n) \equiv \sum_{i=1}^n f_i(x_i), \label{Global_cost_function}  \\
\text{s.t.}~ & x_i = x_j, \forall v_i, v_j, \in \mathcal{V}, \label{constr_same_x}  \\
       & x_i^{[0]} \in \mathcal{X} \subset \mathbb{Q}_{\geq 0}, \forall v_i \in \mathcal{V}, \label{constr_x_in_X} \\
       & b\left(s_i^{[k]}\right) \in \mathcal{B}. \label{constr_quant} 
\end{align} 
\end{subequations}
In \textbf{P1}, $\mathcal{X}$ represents the set of feasible parameter values. 
Also, let $x^*$ denotes the optimal solution to the optimization problem. 
Eq.~\eqref{constr_quant} indicates that during step $k$ of the optimization nodes communicate by transmitting (and, hence, receiving) binary representations (obtained by coding scheme $b:\mathcal{Q}\mapsto \mathcal{B}$) of the quantized values $s_i^{[k]}$ due to bandwidth constraints imposed by the communication channels. 

\noindent
\textbf{Optimization Problem~P2.}
Our second objective is to design a distributed algorithm that enables nodes to collaboratively find an approximate solution to the following optimization problem:
\begin{subequations}
\begin{align}
\mathbf{P2:}~\min_{x\in \mathcal{X}}~ & F(x_1, x_2, ..., x_n) \equiv \sum_{i=1}^n f_i(x_i), \label{Global_cost_function_2}  \\
\text{s.t.}~ & \text{same as} \ \eqref{constr_same_x}, \\
       & \text{same as} \ \eqref{constr_x_in_X}, \\
       & \text{same as} \ \eqref{constr_quant} , \\ 
       & \text{if} \ \| f_i(x_i^{[\gamma_{\beta-1}]}) - f_i(x_i^{[\gamma_\beta]}) \| \leq \varepsilon_{s_1}, \ \text{and} \nonumber \\
       & \text{if} \ \| \nabla f_i(x_i^{[\gamma_\beta]}) \| \leq \varepsilon_{s_2}, \ \forall v_i \in \mathcal{V}, \nonumber \\
       & \text{then terminate operation, } \label{constr_stop} 
\end{align} 
\end{subequations}
where $\beta \in \mathbb{N}$, and $\gamma_\beta$ is the optimization convergence point for which we have $f_i(x_i^{[1 + \gamma_{\beta}]}) = f_i(x_i^{[\gamma_\beta]})$ $\forall v_i \in \mathcal{V}$. 
Eq.~\eqref{constr_stop} means that nodes are tracking the improvement of their local cost function and the norm of their gradient.
Specifically, they are tracking (i) the improvement of their local functions between two consecutive convergence points $\gamma_{\beta + 1}$ and $\gamma_{\beta}$, and (ii) the norm of their gradient at the most recent convergence point. 
If the improvement of the local cost function of every node is less than a predefined threshold $\varepsilon_{s_1}$, and the norm of the gradient at this point is smaller than a predefined threshold $\varepsilon_{s_2}$, then they decide to stop their operation in a distributed way. 

\noindent
\textbf{Optimization Problem~P3.}
Our third objective is to design a distributed algorithm that enables nodes to collaboratively find an approximate solution to the following optimization problem: 
\begin{subequations} 
\begin{align} 
\mathbf{P3:}~\min_{x \in \mathcal{X}}~ & F(x_1, x_2, ..., x_n) \equiv \sum_{i=1}^n f_i(x_i), \label{Global_cost_function_3}  \\ 
\text{s.t.}~ & \text{same as} \ \eqref{constr_same_x}, \\ 
       & \text{same as} \ \eqref{constr_x_in_X}, \\ 
       & b^{\text{FLC}}\left(s_i^{[k]}\right) \in \mathcal{B}^{3}, \label{constr_quant_p3} \\
       & \text{same as} \ \eqref{constr_stop} . 
\end{align} 
\end{subequations} 
Eq.~\eqref{constr_quant_p3} means that nodes are transmitting and receiving $N$-bit quantized values with their neighbors. 

\begin{remark}[Challenges of \textbf{P1, P2, P3}]
    Solving problems \textbf{P1}, \textbf{P2}, and \textbf{P3} poses significant challenges. 
    In \textbf{P1}, the challenge lies in the communication constraints that limit the transmission of nodes' exact states. 
    This requires the design of robust algorithms capable of operating under quantized communication among nodes. 
    Additionally, in \textbf{P2}, nodes must collaboratively monitor the improvement of their local cost functions to determine when to terminate the distributed optimization process. 
    This demands the design of additional coordination mechanisms.  
    Finally, in \textbf{P3}, the $N$-bit communication constraint necessitates the design of algorithms that can distributively adjust the quantizer parameters to maintain optimization accuracy. 
    Addressing the challenges of \textbf{P1, P2, P3} is crucial for addressing optimization problems in resource-constrained and imperfect communication environments.
\end{remark}

\subsection{Assumptions} 

For the development of our results in this paper we make the following assumptions. 

\begin{assum}\label{str_conn}
The network is modeled as a \textit{strongly connected} digraph $\mathcal{G}$. 
\end{assum}

Assumption~\ref{str_conn} is needed to ensure that information propagates among every node in the network (since there is a path from every node to every other node). 
This enables nodes to compute the optimal solution $x^*$ for our problems.

\begin{assum}\label{lipsch_str_conv}
    For every node $v_i$, the local cost function $f_i(x)$ is smooth and strongly convex. 
    This means that for every node $v_i$, for every $x_1, x_2 \in \mathcal{X}$, 
    \begin{itemize}
        \item there exists positive constant $L_i$ such that
        \begin{equation}\label{lipschitz_defn}
            \| \nabla f_i(x_1) - \nabla f_i(x_2) \|_2 \leq L_i \| x_1 - x_2 \|_2, 
        \end{equation}
        \item there exists positive constant $\mu_i$ such that 
        \begin{equation}\label{str_conv_defn}
             f_i(x_2) \geq f_i(x_1) + \nabla f_i(x_1)^\top (x_2 - x_1) + \frac{\mu_i}{2} \| x_2 - x_1 \|_2^2. 
        \end{equation}
    \end{itemize}
\end{assum}

As a consequence of Assumption~\ref{lipsch_str_conv}, the global cost function $F$ (see e.g., \eqref{Global_cost_function}) has Lipschitz-continuity and strong-convexity constants denoted by $L$ and $\mu$, respectively. 
These constants can be bounded as follows: $L \leq \sum_i L_i$, and $\mu \geq \min_i \mu_i$ (see, e.g., \cite[Theorem~15]{JMLR:v20:19-543}).
Assumption~\ref{lipsch_str_conv}, in \eqref{lipschitz_defn} ensures existence of a global optimal solution $x^*$ for \eqref{Global_cost_function}, and that nodes have the capability to compute it. 
This assumption is standard in first-order distributed optimization problems (see e.g., \cite{2018:Xie, 2018:Li_Quannan}). 
Additionally, strong-convexity in \eqref{str_conv_defn} (combined with smoothness) enables us to establish linear convergence rate and ensure that the global cost function $F$ in \eqref{Global_cost_function} possesses only one minimum \cite{bubeck2015convex}. 

\begin{assum}\label{diam_quant}
Every node $v_i \in \mathcal{V}$ knows the diameter $D$ of the underlying network (or an upper bound $D' \geq D$), and a common quantization level $\Delta$. 
\end{assum}
  
In Assumption~\ref{diam_quant}, knowledge of the network diameter $D$ enables nodes to determine in finite time whether their calculated optimal solutions are equal (i.e., condition \eqref{constr_same_x} holds), and a common quantization level $\Delta$ enables communication among nodes by exchanging quantized messages (i.e., condition \eqref{constr_quant} holds). 

\section{Distributed Optimization with Quantized Communication}\label{sec:distr_opt_quant}



\subsection{Optimization Algorithm with Quantized Communication}\label{optim_alg_IFAC} 

In this section, we present a distributed optimization algorithm (detailed below as Algorithm~\ref{algorithm_1}) which solves problem \textbf{P1} in Section~\ref{sec:probForm}. 

\noindent
\vspace{-0.5cm}    
\begin{varalgorithm}{1}
\caption{Quantized Averaged Gradient Descent}
\textbf{Input:} A strongly connected directed graph $\mathcal{G}$ with $n = |\mathcal{V}|$ nodes and $m = |\mathcal{E}|$ edges. 
For every node $v_i \in \mathcal{V}$: static step-size $\alpha \in \mathbb{R}$, digraph diameter $D$, initial state $x_i^{[0]}$, local cost function $f_i$, quantization level $\Delta \in \mathbb{Q}$. 
\\
\textbf{Iteration:} For $k=0,1,2,\dots$, each node $v_i \in \mathcal{V}$ does: 
\begin{list4}
    \item[$1)$] performs one gradient descent step $z_i^{[k+1]} =  x_i^{[k]} - \alpha \nabla f_i(x_i^{[k]})$;
    \item[$2)$] updates its local estimate variable (used to calculate the optimal solution) as $x_i^{[k+1]} =$ Algorithm~\ref{algorithm_1a}($q_{\Delta}^s(z_i^{[k+1]}) / \Delta, D, \Delta $); 
\end{list4} 
\textbf{Output:} Each node $v_j \in \mathcal{V}$ calculates $x^*$ which solves problem \textbf{P1} in Section~\ref{sec:probForm}. 
\label{algorithm_1}
\end{varalgorithm}

\noindent
\vspace{-0.3cm}    
\begin{varalgorithm}{2}
\caption{Finite-Time Quantized Average Consensus}
\textbf{Input:} $\rho_i = q_{\Delta}^s(z_i^{[k+1]}) / \Delta, D, \Delta$. 
\\
\textbf{Initialization:} Each node $v_i \in \mathcal{V}$ does the following: 
\begin{list4}
\item[$1)$] Assigns probability $b_{li}$ to each out-neighbor $v_l \in \mathcal{N}^+_i \cup \{v_i\}$, as follows
\begin{align*}
b_{li} = \left\{ \begin{array}{ll}
         \frac{1}{1 + \mathcal{D}_i^+}, & \mbox{if $l = i$ or $v_{l} \in \mathcal{N}_i^+$,} \\
         0, & \mbox{if $l \neq i$ and $v_{l} \notin \mathcal{N}_i^+$;}\end{array} \right. 
\end{align*} 
\item[$2)$] sets $z_i = 2$, $y_i = 2 \ \rho_i$; 
\end{list4} 
\textbf{Iteration:} For $\lambda = 1,2,\dots$, each node $v_i \in \mathcal{V}$, does: 
\begin{list4a}
\item[$1)$] \textbf{if} $\lambda \mod (D) = 1$ \textbf{then} $M_i = \lceil y_i  / z_i \rceil$, $m_i = \lfloor y_i / z_i \rfloor$; 
\item[$2)$] broadcasts its stopping variables $M_i$, $m_i \in \mathbb{N}$ (used to determine whether convergence has been achieved) to every $v_{l} \in \mathcal{N}_i^+$; receives $M_j$, $m_j$ from every $v_{j} \in \mathcal{N}_i^-$; sets $M_i = \max_{v_{j} \in \mathcal{N}_i^-\cup \{ v_i \}} M_j$, \\ $m_i = \min_{v_{j} \in \mathcal{N}_i^-\cup \{ v_i \}} m_j$; 
\item[$3)$] sets $c_i^z = z_i$; 
\item[$4)$] \textbf{while} $c_i^z > 1$ \textbf{do} 
\begin{list4a}
\item[$4.1)$] $c^y_i = \lfloor y_{i} \  / \  z_{i} \rfloor$; 
\item[$4.2)$] sets its mass variables (used to communicate with other nodes by either transmitting or receiving messages) as $y_{i} = y_{i} - c^y_i$, $z_{i} = z_{i} - 1$, and $c_i^z = c_i^z - 1$; 
\item[$4.3)$] transmits $c^y_i$ to randomly chosen out-neighbor $v_l \in \mathcal{N}^+_i \cup \{v_i\}$ according to $b_{li}$; 
\item[$4.4)$] receives $c^y_j$ from $v_j \in \mathcal{N}_i^-$ and sets 
\begin{align}
y_i & = y_i + \sum_{j=1}^{n} w^{[r]}_{\lambda,ij} \ c^y_{j} \ , \\
z_i & = z_i + \sum_{j=1}^{n} w^{[r]}_{\lambda,ij} \ ,
\end{align}
where $w^{[r]}_{\lambda,ij} = 1$ when node $v_i$ receives $c^y_{i}$, $1$ from $v_j$ at time step $\lambda$ (otherwise $w^{[r]}_{\lambda,ij} = 0$ and $v_i$ receives no message at time step $\lambda$ from $v_j$);
\end{list4a} 
\item[$5)$] \textbf{if} $\lambda \mod D = 0$ \textbf{and} $M_i - m_i \leq 1$ \textbf{then} sets $x_i^{[k+1]} = m_i \Delta$ and stops operation. 
\end{list4a}
\textbf{Output:} $x_i^{[k+1]}$. 
\label{algorithm_1a} 
\end{varalgorithm}

The intuition of Algorithm~\ref{algorithm_1} is the following. 
Initially, each node $v_i$ maintains an estimate $x_i^{[0]}$ of the optimal solution. 
At each time step $k$, each node updates the estimate of the optimal solution by performing a gradient descent step towards the negative direction the node's gradient; see Iteration Step~$1$. 
Then, each node utilizes a finite time quantized averaging algorithm named Algorithm~\ref{algorithm_1a}; see Iteration Step~$2$. 
Algorithm~\ref{algorithm_1a} converges in finite time and enables each node to update its estimate of the optimal solution to be equal to the quantized average of every node’s estimate. 
Note that during Algorithm~\ref{algorithm_1a} the utilized quantization level (i) is the same for every node, (ii) enables quantized communication between nodes, and (iii) determines the desired precision of the solution (as shown in Section~\ref{sec:results}). 
Algorithm~\ref{algorithm_1a} runs between every two consecutive optimization steps $k$ and $k + 1$ of Algorithm~\ref{algorithm_1}. 
For this reason it uses a different time index $\lambda$ (and not $k$ as Algorithm~\ref{algorithm_1}). 
The intuition of Algorithm~\ref{algorithm_1a} is explained below. 

Algorithm~\ref{algorithm_1a} utilizes (i) quantization, (ii) quantized averaging, and (iii) a stopping strategy. 
Specifically, each node $v_i$ utilizes a symmetric mid-rise quantizer with infinite range to quantize its state; see Input. 
During the iteration, at each time step $\lambda$, each node $v_i$ checks if $z_i > 1$. 
If this condition is met the node partitions $y_i$ into $z_i$ equal quantized segments, with some segments potentially containing a higher value than others by one. 
It selects a segment with the lowest $y$-value and sends it to itself, while the other $z_i-1$ segments are transmitted to randomly chosen out-neighbors or to itself. 
After receiving values $y$ and $z$ from its in-neighbors, it adds them to its stored $y_i$ and $z_i$ values and repeats the process. 
The idea is that $\sum_{i=1}^n z_i$ tokens (each having a value $y_i$) perform random walk in the network. 
Each time two or more tokens ``meet'' at a node, their $y$ values become equal. 
Ultimately, this process leads to all $y$ values becoming equal to the quantized average; see Iteration Step~$4$. 
Finally, every $D$ time steps nodes perform a max and min consensus operation; see Iteration Steps~$1, 2$. 
If the stopping condition holds, it scales the solution according to the quantization level and stops the operation; see Iteration Step~$5$. 

\begin{remark}[Advantages of Algorithm~\ref{algorithm_1}]
    The primary characteristic of Algorithm~\ref{algorithm_1} is the varying number of communication and local optimization steps performed by nodes in each iteration. 
    This is inspired from $DGD^t$ in \cite{2019:Wei_Berahas_optimization}. 
    In this algorithm nodes perform $t$ consensus steps for every local optimization step during every iteration. 
    This enables them to operate more efficiently in environments where communication is much cheaper than computation. 
    This approach marks a significant departure from standard distributed optimization algorithms \cite{2009:Nedic_Optim} that involve a single round of communication and local optimization per iteration.
    Algorithm~\ref{algorithm_1} aligns with the approaches \cite{2024:Khatana_Salapaka_optimization, 2022:Jiang_Charalambous, 2023:Bastianello_Rikos_Johansson}. 
    These algorithms involve each node executing a local optimization step followed by a finite-time consensus algorithm during every iteration.
    As we can see in \cite[Fig.~$2$]{2024:Khatana_Salapaka_optimization}, and \cite[Fig.~$1$]{2023:Bastianello_Rikos_Johansson} this approach significantly improves convergence speed in terms of optimization steps. 
    However, compared to \cite{2024:Khatana_Salapaka_optimization, 2022:Jiang_Charalambous, 2023:Bastianello_Rikos_Johansson}, during Algorithm~\ref{algorithm_1} each node employs a finite-time quantized consensus algorithm after the local optimization step. 
    This means that during each iteration the exchange of quantized messages among nodes further enhances communication efficiency while maintaining comparable convergence performance as we show in Section~\ref{sec:results}. 
\end{remark}

\subsection{Convergence of Algorithm~\ref{algorithm_1}}\label{alg_convergence}

We now proceed to analyze the convergence of Algorithm~\ref{algorithm_1}. 
Initially, we establish a bound on the differences between various parameters of Algorithm~\ref{algorithm_1} through a lemma. 
Subsequently, we demonstrate via a theorem that during the execution of Algorithm~\ref{algorithm_1} the state $x_i^{[k]}$ of each node $v_i \in \mathcal{V}$ converges at a linear rate to a neighborhood of the optimal solution $x^*$, thereby addressing problem \textbf{P1}. 

For convenience we rewrite the Iteration steps~$1$ and $2$ of Algorithm~\ref{algorithm_1} are: 
\begin{align} 
z_i^{[k+1]} =&  x_i^{[k]} - \alpha \nabla f_i(x_i^{[k]}), \label{x_medium_update} \\ 
 x_i^{[k+1]} =& \text{Algorithm}~\ref{algorithm_1a}(z_i^{[k+1]} , D, \Delta ).\label{x_final_update} 
\end{align} 
In~\eqref{x_final_update}, after convergence of Algorithm~\ref{algorithm_1a} each node's state is 
\begin{equation}\label{algorithm1a_ac} 
x_i^{[k+1]} = \frac{1}{n} \Bigl ( \sum_{i=1}^{n}\Delta \Bigl \lfloor \frac{z_i^{[k+1]} }{\Delta} \Bigr \rfloor \Bigr ) - \varrho_i^{[k+1]}, \ 0 \le \varrho_i^{[k+1]} \le \Delta, 
\end{equation} 
for every $k \ge 0$, where $\varrho_i^{[k+1]}$ represents the cumulative error incurred by Algorithm~\ref{algorithm_1a} while computing the quantized average of the nodes' initial quantized states at step $k$. 
From the asymmetric quantizers in Section~\ref{sec:preliminaries}, we have 
\begin{align}
  x_i^{[k]} =&  \Delta \Bigl \lfloor \frac{x_i^{[k]}}{\Delta} \Bigr \rfloor + \epsilon_i^{[k]}, \ \text{where} \ 0 \le \epsilon_i^{[k]} \le \Delta, \label{eq_epsilon}
\end{align}
for every $v_i \in \mathcal{V}$, where $\epsilon_i^{[k]}$ is the error due to applying asymmetric quantization to the value $x_i^{[k]}$ at time step $k$.
Let us now denote 
\begin{equation}\label{cumulative_x_z}
\hat{z}^{[k+1]}  \coloneqq \frac{1}{n}\sum_{i=1}^{n} z_i^{[k+1]}, \ \text{and} \ \hat{x}^{[k+1]} \coloneqq \frac{1}{n}\sum_{i=1}^{n} x_i^{[k+1]}, \ k\ge 0 .
\end{equation}
Combining \eqref{eq_epsilon} and \eqref{cumulative_x_z}, we have 
\begin{align}
 \hat{z}^{[k+1]} - x_i^{[k+1]} =& \frac{1}{n}\sum_{i=1}^{n} z_i^{[k+1]} - \frac{1}{n}\sum_{i=1}^{n}\Delta \Bigl \lfloor \frac{z_i^{[k+1]} }{\Delta} \Bigr \rfloor +\varrho_i^{[k+1]}\nonumber \\=& \frac{1}{n}\sum_{i=1}^{n}\epsilon_i^{[k+1]}+\varrho_i^{[k+1]} \le 2 \Delta.\label{x_hatz}
\end{align}

\begin{lemma}\label{lemma_uk_phik}
	For $ k \ge 1 $, the following inequalities hold:
	\begin{align}
	\| \hat x^{[k]}-\hat{z}^{[k]}  \|\le & 2\Delta, \label{uk_hatuk}\\
	\| x_i^{[k]}- \hat x^{[k]} \| \le& 4\Delta. \label{phik}
	\end{align}
\end{lemma}
\begin{proof}
For proving~\eqref{uk_hatuk}, from~\eqref{algorithm1a_ac} we have 
\begin{align}
\| \hat x^{[k]}-\hat{z}^{[k]}  \| =& \|\frac{1}{n}\sum_{i=1}^{n} (x_i^{[k]}-z_i^{[k]}) \| \nonumber\\ =&\|\frac{1}{n}\sum_{i=1}^{n}( - \varrho_i^{[k]} - z_i^{[k]}+ \frac{1}{n}\sum_{i=1}^{n} \Delta \Bigl \lfloor \frac{z_i^{[k]} }{\Delta} \Bigr \rfloor)  \|\nonumber\\
\le&\|\frac{1}{n}\sum_{i=1}^{n} \varrho_i^{[k]}  \| + \|\frac{1}{n}\sum_{i=1}^{n}(  \Delta \Bigl \lfloor \frac{z_i^{[k]} }{\Delta} \Bigr \rfloor -z_i^{[k]})  \| \nonumber\\ \le& 2\Delta .
\end{align}

For proving~\eqref{phik}, from~\eqref{x_hatz} and \eqref{uk_hatuk}, we have 
\begin{equation*}
\| x_i^{[k]}- \hat x^{[k]} \| \le \| x_i^{[k]}- \hat{z}^{[k]} \| + \| \hat{z}^{[k]}- \hat x^{[k]} \| \le 4\Delta .
\end{equation*}
\end{proof}
Let us denote 
$$ 
u^{[k]} := \sum_{i=1}^{n} \nabla f_i(x_i^{[k]}), \ \hat u^{[k]} := \sum_{i=1}^{n} \nabla f_i(\hat x^{[k]}). 
$$
Combining Assumption~\ref{lipsch_str_conv} with~\eqref{phik}, we have 
\begin{align} 
&\| \nabla f_i(x_i^{[k]}) - \nabla f_i(\hat x^{[k]}) \| \le  L_i \| x_i^{[k]} -\hat  x^{[k]}\|\le  L_i 4 \Delta \\
&\|u^{[k]} - \hat u^{[k]} \| \le  \sum_{i=1}^{n}\| \nabla f_i(x_i^{[k]})-\nabla f_i(\hat{x}^{[k]}) \| \nonumber \\ &\le  \sum_{i=1}^{n} L_i \| x_i^{[k]} -\hat  x^{[k]}\| \le L 4 \Delta, \ \text{where} \ L = \sum_i L_i. \label{uk_hat}
\end{align} 
Let us now denote $ e^{[k]} = \hat{x}^{[k]}-\hat{z}^{[k]}$, and $\hat \alpha := \alpha / n$.
From~\eqref{uk_hatuk} we have $ \|e^{[k]}\|\le 2\Delta $. 
From \eqref{x_medium_update} we get
\begin{align}
\hat{x}^{[k+1]} =& \hat{z}^{[k+1]} + e^{[k+1]}\nonumber\\
=& \frac{1}{n}\sum_{i=1}^{n} (x_i^{[k]} - \alpha \nabla f_i(x_i^{[k]})) + e^{[k+1]}\nonumber\\
=&\hat{x}^{[k]} - \hat \alpha u^{[k]} + e^{[k+1]},  \ \text{for} \ k\ge 1. \label{hatx_calcu2}
\end{align}
We now establish our algorithm's convergence rate, determine an upper bound for the step-size via the following theorem. 

\begin{theorem}\label{theorem_convergence_stronglyConvex} 
    Under Assumptions~\ref{str_conn}--\ref{diam_quant}, when the step-size $ \alpha $ satisfies $ \alpha \in (0, \frac{2n}{\mu + L}] $, where  
    $L =  \sum_i L_i$, and $\mu = \min_i \mu_i$, 
    Algorithm~\ref{algorithm_1} generates a sequence of points $ \{x^{[k]}\}$ (i.e., the variable $x_i^{[k]}$ of each node $v_i \in \mathcal{V}$) which satisfy
	\begin{align}\label{linear_convergence}
	\| \hat{x}^{[k+1]} - x^{*}\|
	\le
	(1-\frac{\alpha \mu}{n})\| \hat{x}^{[k]} - x^{*}  \|  + \mathcal{O}(\Delta) ,
	\end{align}
	where $\Delta$ is the quantizer level and
		\begin{align}
		 \mathcal{O}(\Delta) 
		 =&( \frac{4\alpha L}{n} + 2)\Delta. \label{throrem1_b} 
		\end{align}
\end{theorem}

\begin{proof}
See Appendix~\ref{appen_theorem}. 
\end{proof} 

Theorem~\ref{theorem_convergence_stronglyConvex} shows that Algorithm~\ref{algorithm_1} converges linearly to a neighborhood of the optimal solution. 
This neighborhood is determined by the quantization level $\Delta$ as shown in Section~\ref{sec:results}. 

%
%
%
%
\section{Solution Enhancement and Stopping via Event-Triggered Zooming}\label{sec:distr_grad_zoom_quant}

Inspired by strategies for terminating the gradient methods according to performance criteria \cite[Section~$1.2$]{2016:nonlinear_program_Bertsekas} and the need for achieving high precision solutions, in this section we propose a novel distributed algorithm (named below as Algorithm~\ref{algorithm_2}). 
Our algorithm aims to approach the optimal solution as closely as possible within the limitations of quantized communication, while incorporating a distributed stopping criterion that allows nodes to collectively determine when to terminate operations based on specific performance metrics. 

\subsection{Optimization Algorithm with Event-Triggered Zooming over Quantized Communication}\label{optim_alg_IFAC} 

We now present a distributed algorithm which solves problem \textbf{P2} in Section~\ref{sec:probForm}. 

\begin{varalgorithm}{3}
\caption{Quantized Averaged Gradient Descent with Zoomed Communication}
\textbf{Input:} A strongly connected directed graph $\mathcal{G}$ with $n = |\mathcal{V}|$ nodes and $m = |\mathcal{E}|$ edges. 
Static step-size $\alpha \in \mathbb{R}$, digraph diameter $D$, initial value $x_i^{[0]}$, local cost function $f_i$, error bounds $\varepsilon_{s_1}$ and $\varepsilon_{s_2}$, quantization level $\Delta \in \mathbb{Q}$, refinement constant $c_r \in \mathbb{N}$, for every node $v_i \in \mathcal{V}$. 
\\ 
\textbf{Initialization:} Each node $v_i \in \mathcal{V}$ sets $\nu_{\text{total}} = 0$, $\beta = \nu_{\text{total}}$, $S_i = \{ 0 \}$. \\ 
\textbf{Iteration:} For $k = 0,1,2,\dots$, each node $v_i \in \mathcal{V}$ does the following: 
\begin{list4} 
\item[$1)$] performs one gradient descent step $z_i^{[k+1]} =  x_i^{[k]} - \alpha \nabla f_i(x_i^{[k]})$; 
\item[$2)$] updates its local estimate variable (used to calculate the optimal solution) as $x_i^{[k+1]} =$ Algorithm~\ref{algorithm_1a}($q_{\Delta}^s(z_i^{[k+1]}) / \Delta, D, \Delta$); 
\item[3)] \textbf{if} $x_i^{[k+1]} = x_i^{[k]}$, \textbf{then} 
\begin{list4a} 
\item[$3a)$] set $\nu_{\text{total}} = \nu_{\text{total}} + 1$, $\beta = \nu_{\text{total}}$, and $\gamma_{\beta} = k$ (where $\nu_{\text{total}}$ is an indicator of the number of times the algorithm has converged to a neighborhood, and $\gamma_\beta$ is the time step at which nodes have converged to a neighborhood of the optimal solution); 
\item[$3b)$] set $S_i = S_i \cup \{ \gamma_{\beta} \}$ (where $S_i$ is the set in which all the time steps at which nodes have converged to a neighborhood of the optimal solution are stored); 
\item[$3c)$] \textbf{if} $\| f_i(x_i^{[\gamma_{\beta-1}]}) - f_i(x_i^{[\gamma_\beta]}) \| \leq \varepsilon_{s_1}$, \textbf{and} $\| \nabla f_i(x_i^{[\gamma_\beta]}) \| \leq \varepsilon_{s_2}$ \textbf{then} set $\text{vot}_i = 0$; 
\item \textbf{else} set $\text{vot}_i = 1$; 
\item[$3d)$] $\text{flag}_i$ = max - Consensus ($\text{vot}_i$); 
\item[$3e)$] \textbf{if} $\text{flag}_i = 0$ \textbf{then} terminate operation (where $\text{flag}_i$ is used to decide whether to terminate the optimization algorithm operation); 
\item \textbf{else} set $\Delta = \Delta / c_r$ and go to Step~$1$;
\end{list4a} 
\end{list4} 
\textbf{Output:} Each node $v_i \in \mathcal{V}$ calculates $x_i^*$ which solves problem \textbf{P2} in Section~\ref{sec:probForm}. 
\label{algorithm_2} 
\end{varalgorithm} 

The intuition of Algorithm~\ref{algorithm_2} is the following. 
During initialization, each node $v_i$ maintains an estimate $x_i^{[0]}$ of the optimal solution, the quantization level $\Delta$, and the refinement constant $c_r$ (utilized for refining the quantization level). 
Let us recall that the quantization level $\Delta$ remains consistent across all nodes, and determines the desired precision level of the solution. 
Additionally, each node stores in a set $S_i$ the optimization time steps when nodes have converged to a neighborhood of the optimal solution. 
At each time step $k$, each node updates the estimate of the optimal solution through a gradient descent step, and then employs the finite-time quantized averaging algorithm Algorithm~\ref{algorithm_1a}; see Iteration Steps~$1$, $2$. 
This enables every node to satisfy \eqref{constr_quant} and compute an estimate of the optimal solution $x_i^{[k+1]}$ within finite time, also satisfying \eqref{constr_same_x}. 
Next, each node $v_i$ checks whether the computed estimate of the optimal solution $x_i^{[k+1]}$ is equal to the previous optimization step $x_i^{[k]}$; see Iteration Step~$3$.
This criterion indicates that nodes reached the optimization convergence point for the current quantization level (i.e., a neighborhood that depends on the current quantization level). 
In this case, each node $v_i$ stores the corresponding time step $k$ at the set $S_i$; see Iteration Step~$3b$. 
Then, it checks (i) if the difference between its local function's value at the current optimization convergence point $f_i(x_i^{[\gamma_{\beta}]})$ and that at the previous optimization convergence point $f_i(x_i^{[\gamma_{\beta-1}]})$ falls below a designated threshold $\varepsilon_{s_1}$, and (ii) whether the norm of the gradient at the current optimization convergence point is below a specified threshold $\varepsilon_{s_2}$; see Iteration Step~$3c$. 
If both conditions hold (the difference and the norm), then it sets its voting variable equal to $0$, otherwise it sets it to $1$; see Iteration Step~$3c$. 
Then nodes are performing a max-Consensus protocol to decide whether they will continue the operation of Algorithm~\ref{algorithm_2}; outlined in Iteration Step~$3d$. 
Specifically, if both conditions hold for every node, then nodes decide to stop the operation of Algorithm~\ref{algorithm_2}; see Iteration Step~$3e$. 
Otherwise, nodes refine the quantization level with the refinement constant $c_r$, and repeat the operation of the algorithm; see Iteration Step~$3e$. 

\begin{remark}[Operation of Algorithm~\ref{algorithm_2}]
    Algorithm~\ref{algorithm_2} incorporates a mechanism for distributed enhancement of optimal solution precision. 
    This mechanism adapts to desired accuracy levels while maintaining quantized communication. 
    Furthermore, it enables nodes to collectively determine when to terminate operations based on specific performance criteria. 
    In other distributed optimization algorithms, convergence either to a neighborhood or to the exact solution depends on specific tuning of the stepsize. 
    In our case, Algorithm~\ref{algorithm_2} incorporates a distributed mechanism to implement the stopping criterion of \cite[Section~$1.2$]{2016:nonlinear_program_Bertsekas} along with a second stopping criterion that depends on the overall improvement of the local cost functions. 
    In the literature there are also other strategies designed to terminate the operation of gradient descent in a centralized manner \cite[Section~$3.5$]{nocedal1999numerical} and \cite[Section~$3.3.1$]{boyd2011distributed}. 
    However, to the authors' knowledge the distributed termination of optimization algorithms based on performance criteria remains largely unexplored, with only \cite{ayken2015diffusion} presenting a related diffusion-based stopping criterion. 
    As a result, Algorithm~\ref{algorithm_2} not only enhances precision and efficiency but also addresses a significant gap in the literature by providing a robust distributed stopping mechanism based on performance criteria.
\end{remark}

\subsection{Convergence of Algorithm~\ref{algorithm_2}}\label{alg_convergence_2} 

We now analyze the convergence of Algorithm~\ref{algorithm_2} via the following theorem. 
Its intuition shares similarities with Theorem~\ref{theorem_convergence_stronglyConvex}. 
However, we present a sketch of the proof for completeness. 

\begin{theorem}\label{theorem_convergence_stronglyConvex_2} 
    Under Assumptions~\ref{str_conn}--\ref{diam_quant}, when the step-size $ \alpha $ satisfies $ \alpha \in (0, \frac{2n}{\mu + L}] $ where  
    $L =  \sum_i L_i$, and $\mu = \min_i \mu_i$, 
    Algorithm~\ref{algorithm_2} generates a sequence of points $ \{x^{[k]}\}$ (i.e., the variable $x_i^{[k]}$ of each node $v_i \in \mathcal{V}$) which satisfy
	\begin{align}\label{linear_convergence_2}
	\| \hat{x}^{[k+1]} - x^{*}\|
	\le
	(1-\frac{\alpha \mu}{n})\| \hat{x}^{[k]} - x^{*} \|   + \mathcal{O}(\Delta) ,
	\end{align}
	where $\Delta$ is the quantizer level and
		\begin{align}
		 \mathcal{O}(\Delta) 
		 =&( \frac{4\alpha L}{n} + 2)\Delta. \label{throrem1_b_2} 
		\end{align}
\end{theorem}

\begin{proof}
Let us suppose that during the operation of Algorithm~\ref{algorithm_2}, for time steps $[0, k_1)$ the condition in Iteration-Step~$3$ does not hold (i.e., $x_i^{[k+1]} < x_i^{[k]}, \forall k \in [0, k_1)$). 
In this case, for for time steps $[0, k_1)$, Algorithm~\ref{algorithm_2} exhibits linear convergence rate towards the optimal solution of problem \textbf{P2} (as described in Theorem~\ref{theorem_convergence_stronglyConvex}). 
Let us suppose that during time step $k_1$ the condition in Iteration-Step~$3$ holds (i.e., $k_1$ is a convergence point).
In this case, nodes check (i) if the improvement of their local cost functions at this convergence point compared to the previous convergence point is below a threshold, and (ii) if the norm of their gradient is below a threshold. 
If both conditions hold for every node, then nodes decide to stop their operation. 
If at least one of these two conditions do not hold then, the quantization level is refined. 
Then, for a set of subsequent time steps, Algorithm~\ref{algorithm_2} continues its operation converging linearly towards the optimal solution. 
As a consequence, the operation of Algorithm~\ref{algorithm_2} enables nodes to converge to the optimal solution linearly (as described in Theorem~\ref{theorem_convergence_stronglyConvex}) with the difference being that at some time steps the optimization process stops and restarts with a refined quantization level $\Delta$. 
\end{proof} 

\begin{remark}[Extending Algorithm~\ref{algorithm_2}]
In Algorithm~\ref{algorithm_2}, the stopping criterion is based on the overall improvement of the local cost functions and the norm of the gradient of each node's local cost function at the convergence point. 
However, these conditions can be extended to accommodate more complex scenarios. 
For example, nodes can also evaluate the difference in the norms of the local cost functions at two consecutive convergence points or consider the normalized norms of the gradients \cite[Section~1.2]{2016:nonlinear_program_Bertsekas}. 
As a result, Algorithm~\ref{algorithm_2} provides a flexible framework for implementing distributed stopping and refinement strategies according to the specific conditions dictated by the corresponding application.
\end{remark}

\section{Solution Enhancement and Stopping with $N$-bit Quantizers}\label{sec:distr_grad_zoom_quant_3_bit} 

In the previous section, Algorithm~\ref{algorithm_2} exhibits quantized communication among nodes enhancing operational efficiency. 
However, as the quantization level is refined, the number of bits the nodes are required to exchange increases (since nodes utilize a symmetric mid-rise quantizer with infinite range). 
This, in turn, reduces the communication efficiency benefits provided by quantization.
In realistic scenarios (e.g., wireless sensor networks or Internet of Things devices \cite{2022:Prakash_Miao_IoT, 2022:Jinjin_Cheng_FedLearning}) communication among nodes is often constrained to only a few bits due to bandwidth limitations. 
In this section we present an efficient distributed algorithm that can operate effectively under such constraints.  

\subsection{Distributed Optimization Algorithm with $N$ Bit Adaptive Quantization}\label{distr_alg}

We now present our proposed algorithm detailed below as Algorithm~\ref{algor_4}. 

\begin{varalgorithm}{4}
\caption{Quantized Averaged Gradient Descent with Finite Bit Adaptive Quantization} 
\textbf{Input:} A strongly connected directed graph $\mathcal{G}$ with $n = |\mathcal{V}|$ nodes and $m = |\mathcal{E}|$ edges. 
Every node $v_i \in \mathcal{V}$ has: static step-size $\alpha \in \mathbb{R}$, digraph diameter $D$, initial value $x_i^{[0]}$, local cost function $f_i$, quantization level $\Delta^{[0]} \in \mathbb{Q}$, zoom-in constant $c^{\text{in}} \in \mathbb{Q}_{> 1}$, zoom-out constant $c^{\text{out}} \in \mathbb{Q}_{> 1}$. 
\\
\textbf{Initialization:} Each node $v_i \in \mathcal{V}$ sets $b_q = 0$, $\nu_{\text{in}} = 0$, $\nu_{\text{out}} = 0$, $\nu_{\text{total}} = 0$. \\ 
\textbf{Iteration:} For $k = 0,1,2, \dots$, each node $v_i \in \mathcal{V}$ does the following: 
\begin{list4}
\item[$1)$] performs one gradient descent step $z_i^{[k+1]} =  x_i^{[k]} - \alpha \nabla f_i(x_i^{[k]})$; 
\item[$2)$] updates its local estimate variable (used to calculate the optimal solution) as $x_i^{[k+1]} =$ Algorithm~\ref{algorithm_1a}($Q^{\text{$N$MRU}}(b_q, z_i^{[k+1]}, \Delta^{[\nu_{\text{total}}]}) / \Delta^{[\nu_{\text{total}}]}$, $D, \Delta^{[\nu_{\text{total}}]}$); 
\item[3)] \textbf{if} $x_i^{[k+1]} = x_i^{[k]}$, \textbf{then} 
\begin{list4a} 
\item[$3a)$] set $\nu_{\text{total}} = \nu_{\text{total}} + 1$, $\beta = \nu_{\text{total}}$, $\gamma_{\beta} = k$ (where $\nu_{\text{total}}$ is an indicator of the number of times the algorithm has converged to a neighborhood, and $\gamma_\beta$ is the time step at which nodes have converged to a neighborhood of the optimal solution); 
\item[$3b)$] set $S_i = S_i \cup \{ \gamma_{\beta} \}$ (where $S_i$ is the set in which all the time steps at which nodes have converged to a neighborhood of the optimal solution are stored); 
\item[$3c)$] \textbf{if} $\| f_i(x_i^{[\gamma_{\beta-1}]}) - f_i(x_i^{[\gamma_\beta]}) \| \leq \varepsilon_{s_1}$, \textbf{and} $\| \nabla f_i(x_i^{[\gamma_\beta]}) \| \leq \varepsilon_{s_2}$ \textbf{then} set $\text{vot}_i = 0$; 
\item \textbf{else} set $\text{vot}_i = 1$; 
\item[$3d)$] $\text{flag}_i$ = max - Consensus ($\text{vot}_i$); 
\item[$3e)$] \textbf{if} $\text{flag}_i = 0$ \textbf{then} terminate operation (where $\text{flag}_i$ is used to decide whether to terminate the optimization algorithm operation); 
\item \textbf{else if} $\text{flag}_i = 1$ \textbf{then}  
\begin{list4a}
\item[$3e.1)$] \textbf{if} $x_i^{[k+1]} \geq b_q + N \Delta^{[\nu_{\text{total}}]}$, \textbf{or} $x_i^{[k+1]} < b_q - N \Delta^{[\nu_{\text{total}}]}$ \textbf{then} set 
\begin{subequations} 
\begin{align} 
    & \nu_{\text{out}} = \nu_{\text{out}} + 1, \label{zoom_out_0} \\ 
    & b_q = x_i^{[k+1]},\label{zoom_out_1} \\
    & \Delta^{[\nu_{\text{total}}]} := c^{\text{out}} \Delta^{[\nu_{\text{total}}-1]},\label{zoom_out_2} 
\end{align} 
\end{subequations} 
(where $\nu_{\text{out}}$ is an indicator for the number of Zoom-out for the utilized quantizer), 
\item[$3e.2)$] \textbf{else} set 
\begin{subequations} 
\begin{align} 
    & \nu_{\text{in}} = \nu_{\text{in}} + 1, \label{zoom_in_0} \\ 
    & b_q = x_i^{[k+1]},\label{zoom_in_1} \\
    & \Delta^{[\nu_{\text{total}}]} := \Delta^{[\nu_{\text{total}}-1]} / c^{\text{in}},\label{zoom_in_2} 
\end{align} 
\end{subequations} 
(where $\nu_{\text{in}}$ is an indicator for the number of Zoom-in for the utilized quantizer), 
\end{list4a} 
\end{list4a} 
\item[$4)$] go to Step~$1$;
\end{list4} 
\textbf{Output:} Each node $v_i \in \mathcal{V}$ calculates the optimal solution $x_i^*$ of problem \textbf{P3} in Section~\ref{sec:probForm}. 
\label{algor_4} 
\end{varalgorithm}

The intuition of Algorithm~\ref{algor_4} can be summarized in the following three parts. 
\\ \noindent 
\textit{Part~$1$: Gradient Descent and Coordination.} 
During each time step $k$, each node $v_i$ updates its estimate of the optimal solution by performing a gradient descent step towards the negative direction of the gradient of the node's local cost function; see Iteration Step~$1$. 
Subsequently, each node employs Algorithm~\ref{algorithm_1a}; see Iteration Step~$2$. 
Algorithm~\ref{algorithm_1a} allows nodes to meet the constraint \eqref{constr_quant} and to compute within finite time an estimate of the optimal solution that satisfies \eqref{constr_same_x}.
\\ \noindent 
\textit{Part~$2$: Distributed Stopping.} 
Following the execution of Algorithm~\ref{algorithm_1a}, nodes check if $x_i^{[k+1]} = x_i^{[k]}$ (this condition indicates convergence to a neighborhood of the optimal solution defined by the current quantization level). 
When this condition is met, each node $v_i$ stores the corresponding time step $k$ in the set $S_i$; see Iteration Step~$3b$. 
Subsequently, each node evaluates two conditions: whether the difference between its local function's values at the current and previous optimization convergence points ($f_i(x_i^{[\gamma_{\beta}]})$ and $f_i(x_i^{[\gamma_{\beta-1}]})$) is below a threshold $\varepsilon_{s_1}$, and if the gradient norm at the current point is below a threshold $\varepsilon_{s_2}$; see Iteration Step~$3c$. 
Nodes set their voting variable to $0$ if both conditions are satisfied, or $1$ otherwise. 
A max-Consensus protocol is then performed to determine whether to continue Algorithm~\ref{algor_4}; see Iteration Step~$3d$. 
If all nodes satisfy both conditions, the algorithm terminates (see Iteration Step~$3e$); otherwise, nodes proceed to adjust their quantizer parameters (explained below). 
\\ \noindent
\textit{Part~$3$: Iteration -- Zoom-out or Zoom-in.} 
When nodes determine that they have converged to a neighborhood of the optimal solution but opt to continue the operation of Algorithm~\ref{algor_4}, they proceed to modify the parameters of their quantizer through one of the two following procedures. 
\\ \noindent 
\textit{Part~$3$-A: Iteration -- Zoom-out.} 
If the estimation exceeds the quantizer's dynamic range (i.e., the quantizer has saturated because $x_i^{[k+1]}$ falls within $(-\infty, b_q - N\Delta^{[\nu_{\text{total}}]}) \cup [b_q + N \Delta^{[\nu_{\text{total}}]}, +\infty)$), then the optimal solution lies outside the current quantizer range. 
In this scenario, nodes are unable to estimate the optimal solution with the existing quantizer parameters and must perform a zoom-out operation. 
Zoom-out involves setting the quantizer basis equal to the calculated estimation $x_i^{[k+1]}$ and increasing the quantization level by multiplying it with the zoom-out constant. 
Additionally, nodes increase the zoom-out counter $\nu_{\text{out}}$ to track the number of zoom-out operations performed. 
The primary goal of zoom-out is to adjust the quantizer's parameters to cover a broader range of values, potentially including the optimal solution; see Iteration Step~$3e.1$. 
\\ \noindent 
\textit{Part~$3$-A: Iteration -- Zoom-in.} 
When the estimate falls within the quantizer's dynamic range (i.e, when $x_i^{[k+1]} \in [b_q - N\Delta^{[\nu_{\text{total}}]}, b_q + N\Delta^{[\nu_{\text{total}}]})$), then nodes can estimate the optimal solution using the current quantizer parameters. 
In this scenario, nodes perform a zoom-in operation by setting the quantizer basis to be equal to the calculated estimation $x_i^{[k+1]}$ and reducing the quantization level by dividing it by the zoom-in constant. 
Additionally, nodes increase the zoom-in counter $\nu_{\text{in}}$ to track the number of zoom-in operations performed. 
The purpose of zooming in is to refine the quantizer's parameters, focusing on a narrower range of values that likely contains the optimal solution. 
Within this smaller range, nodes can estimate the solution with greater precision, as the updated $\Delta^{[\nu_{\text{total}}]}$ is smaller than its original value; see Iteration Step~$3e.2$. 

\subsection{Convergence of Algorithm~\ref{algor_4}}\label{algor_4_convergence} 

We now analyze the convergence of Algorithm~\ref{algor_4}. 
First, we introduce a proposition to demonstrate that by performing a zoom-out operation, nodes can adjust the quantizer's dynamic range in finite time to include the optimal solution. 
We then present a theorem establishing the linear convergence rate of our algorithm. 

\begin{prop}[Finite Zoom-out Instances]\label{zoom_out_locate_solution}   
    Let us assume that during the operation of Algorithm~\ref{algor_4} we have $x^* \notin (b_q - N\Delta^{[0]}, b_q + N\Delta^{[0]})$. 
    Then, after a finite number of time steps $\nu_0$ for which 
    \begin{equation}\label{bound_time_steps}
        \nu_0 > \left\lceil \frac{\log(x^*) - \log(N\Delta^{[0]})}{\log(c^{\text{out}})} \right\rceil , 
    \end{equation}
    nodes are able to calculate $\Delta^{[\nu_0]}$ such that $x^* \in (b_q - N\Delta^{[\nu_0]}, b_q + N\Delta^{[\nu_0]})$.
\end{prop} 

\begin{proof} 
    Let us assume that during time step $k=0$ of Algorithm~\ref{algor_4} we have $b_q = 0$, and $x^* \notin (-N\Delta^{[0]}, N\Delta^{[0]})$. 
    In this case, nodes will zoom-out (see Iteration Step~$3e.2$ of Algorithm~\ref{algor_4}). 
    The updated quantization level is equal to $\Delta^{[1]} = c^{\text{out}} \Delta^{[0]}$, and the updated quantizer range (for the case where the quantizer is not saturated) is $(-N\Delta^{[1]}, N\Delta^{[1]})$. 
    Continuing this process (i.e., nodes zoom-out for subsequent time steps $\nu_0 > 1$) we have that nodes will be able to calculate an updated quantization level $\Delta^{[\nu_0]}$ for which $x^* \in (-N\Delta^{[\nu_0]}, N\Delta^{[\nu_0]})$ after
    $\nu_0$
    time steps where $\Delta^{[\nu_0]} = \Delta^{[0]} (c^{\text{out}})^{\nu_0}$ and
    $$
        \nu_0 > \left\lceil \frac{\log(x^*) - \log(N\Delta^{[0]})}{\log(c^{\text{out}})} \right\rceil. 
    $$
    Note that in this proof, without loss of generality, we assumed $b_q$ remains equal to zero during the zoom-out operation (i.e., \eqref{zoom_out_1} is not executed). 
    This means that the calculated bound in \eqref{bound_time_steps} can be further improved. 
    However, the analysis for the case where $b_q$ is updated is more challenging and will be considered as a future direction.
\end{proof}

We now show the linear convergence of Algorithm~\ref{algor_4} via the following theorem. 
The underlying concept is analogous to Theorem~\ref{theorem_convergence_stronglyConvex}, but we provide an outline of the proof for completeness. 

\begin{theorem}\label{converge_algor_4}
Under Assumptions~\ref{str_conn}--\ref{diam_quant}, when the step-size $\alpha$ satisfies $ \alpha \in (0, \frac{2n}{\mu + L}] $ where 
$L =  \sum_i L_i$, and $\mu = \min_i \mu_i$ 
Algorithm~\ref{algor_4} generates a sequence of points $ \{x^{[k]}\} $ (i.e., the variable $x_i^{[k]}$ of each node $v_i \in \mathcal{V}$) which satisfy 
	\begin{align}\label{linear_convergence_algor_4}
	\| \hat{x}^{[k+1]} - x^{*}\|
	\le
	(1-\frac{\alpha \mu}{n})\| \hat{x}^{[k]} - x^{*} \| + \mathcal{O}(\Delta^{[\nu_{\text{total}}]}) ,
	\end{align}
	where $\Delta^{[\nu_{\text{total}}]} $ is the quantizer level and
		\begin{align}
		 \mathcal{O}(\Delta^{[\nu_{\text{total}}]}) 
		 =&( \frac{4\alpha L}{n} + 2)\Delta^{[\nu_{\text{total}}]}. \label{throrem1_b_algor_4} 
		\end{align} 
\end{theorem} 

\begin{proof} 
Let us assume that, during the operation of Algorithm~\ref{algor_4}, for a set of time steps $[k_0, k_1)$ the condition in Iteration-Step~$3$ does not hold (i.e., $x_i^{[k+1]} < x_i^{[k]}, \forall k \in [k_0, k_1)$). 
In this case, for that set of time steps, Algorithm~\ref{algor_4} exhibits linear convergence rate towards the optimal solution as described in Theorem~\ref{theorem_convergence_stronglyConvex}. 
If at time step $k_1$ the condition in Iteration-Step~$3$ holds, then the quantization level is refined and the quantizer basis is adjusted. 
Then, for a set of subsequent time steps, Algorithm~\ref{algor_4} continues its operation converging linearly towards the optimal solution. 
As a consequence, the operation of Algorithm~\ref{algor_4} allows nodes to converge to the optimal solution linearly (as described in Theorem~\ref{theorem_convergence_stronglyConvex}) with the difference being that at some time steps the optimization process stops and restarts with a refined quantization level $\Delta^{[\nu_{\text{total}}]}$. 
Note now that the term $\mathcal{O}(\Delta^{[\nu_{\text{total}}]})$ in \eqref{linear_convergence_algor_4} appears due to nodes utilizing the uniform quantizer $Q^{\text{$N$MRU}}$ during Algorithm~\ref{algor_4}. 
The precision of the optimal solution calculation is impacted by the term $\mathcal{O}(\Delta^{[\nu_{\text{total}}]})$. 
Since during Algorithm~\ref{algor_4} the quantization level $\Delta^{[\nu_{\text{total}}]}$ is decreased infinitely often (i.e., nodes will perform zoom-in operations infinitely often), this means that at some optimization step $k_0$ the condition in Iteration-Step~$3c$ will hold for every node. 
Thus, that after Iteration-Steps~$3d$, $3e$, nodes will decide to terminate the operation of Algorithm~\ref{algor_4}. 
As a result, Algorithm~\ref{algor_4} enables nodes to calculate the optimal solution of problem \textbf{P3} with linear convergence, and terminate their operation while exchanging messages consisting of $N$ bits. 
\end{proof}

\begin{remark}
    In the previous section, Algorithm~\ref{algorithm_2} requires refining the quantization level to enhance the precision of the optimal solution. 
    This results in an increased number of bits exchanged among nodes. 
    In contrast, Algorithm~\ref{algor_4} maintains a fixed number of $N$ exchanged bits for each message. 
    This characteristic is particularly significant in environments where communication bandwidth is limited to only a few bits among nodes (e.g., wireless networks or industrial control systems). 
    Additionally, Algorithm~\ref{algor_4} can approach the optimal solution with precision similar to that of Algorithm~\ref{algorithm_2} despite the fixed-bit communication. 
    Moreover, Algorithm~\ref{algor_4} retains the advantageous feature of Algorithm~\ref{algorithm_2}, allowing nodes to distributively determine when to terminate the operation based on performance guarantees. 
    However, the trade-off is that Algorithm~\ref{algor_4} may require more time steps to converge compared to Algorithm~\ref{algorithm_2} (which is a logical consequence of exchanging fewer bits of information per communication round). 
    This trade-off will be also shown in Section~\ref{sec:results}. 
\end{remark}

\section{Application: Distributed Sensor Fusion for Target Localization}\label{sec:results} 

We now present an application of our proposed algorithms to the problem of distributed sensor fusion for target localization \cite{2020:Xiaohua_Fuwen_Survey, 2022:Doostmohammadian_Houman_TrackingTarget}. 
In this problem, multiple sensors (or nodes) are deployed across a network. 
Each node is equipped with a sensor that provides noisy measurements of the target's position. 
The objective is to estimate the target's position by fusing the local measurements from all sensors in a distributed manner. 
For simplicity, we consider (i) the static distributed sensor fusion problem (though our results can be extended to dynamic scenarios), and (ii) each node has one measurement of the target that is one-dimensional (although the approach can be easily extended to each node having  multiple multi-dimensional measurements). 
In the first part of this section, we compare the operation of Algorithm~\ref{algorithm_1}, Algorithm~\ref{algorithm_2}, and Algorithm~\ref{algor_4} in terms of convergence speed and precision, and total communication requirements (see Section~\ref{comparison_3algs} below). 
In the second part, we compare the operation of our algorithms against other algorithms from the literature in terms of convergence speed (see Section~\ref{sec:Comparisons} below).

\subsection{Operational Comparisons of Algorithm~\ref{algorithm_1}, Algorithm~\ref{algorithm_2}, and Algorithm~\ref{algor_4}}\label{comparison_3algs}

In Fig.~\ref{20_nodes_V4} we demonstrate the operation of our Algorithm~\ref{algorithm_1}, Algorithm~\ref{algorithm_2}, and Algorithm~\ref{algor_4} over a random digraph of $20$ sensor nodes. 
We show the error $e^{[k]}$ defined as 
\begin{equation}\label{eq:distance_optimal}
    e^{[k]} = \sqrt{ \sum_{i=1}^n \frac{(x_i^{[k]} - x^*)^2}{(x_i^{[0]} - x^*)^2} } , 
\end{equation}
in a logarithmic scale against the number of iterations with $x^*$ being the optimal solution of the optimization problems \textbf{P1}, \textbf{P2}, \textbf{P3} (i.e., the position of the target). 
For each node $v_i$ we have: $\alpha = 0.12$, initial estimated position of the target $x_i^{[0]} \in [1, 5]$, a quadratic local cost function of the form $f_i(x_i^{[0]}) = \frac{1}{2} \beta_i (x_i^{[0]} - x_0)^2$ with $x_0$ being the measurement obtained by node $v_i$ (randomly chosen in the set $\{ 1, 2, 3, 4, 5 \}$), and $\beta_i$ randomly chosen in the set $\{ 1, 2, 3, 4, 5 \}$, zoom-in constant $c^{\text{in}} = 4/3$, zoom-out constant $c^{\text{out}} = 2$, refinement constant $c_r=2$, initial quantizer basis $b_q = 0$, and error bounds 
$\varepsilon_{s_1} = 10^{-3}$, $\varepsilon_{s_2} = 10^{-3}$ for Algorithm~\ref{algorithm_2} and $\varepsilon_{s_1} = 10^{-5}$, $\varepsilon_{s_2} = 10^{-5}$ for Algorithm~\ref{algor_4}. 
For Algorithm~\ref{algorithm_1} we have that the quantization level $\Delta \in \{0.1, 0.01, 0.001\}$. 
For Algorithm~\ref{algorithm_2} we have initial quantization level $\Delta = 0.1$. 
For Algorithm~\ref{algor_4} we have initial quantization level $\Delta^{[0]} = 0.1$, and each message consists of $N=3$~bits. 

\begin{figure}[t]
    \centering
    \includegraphics[width=0.9\linewidth]{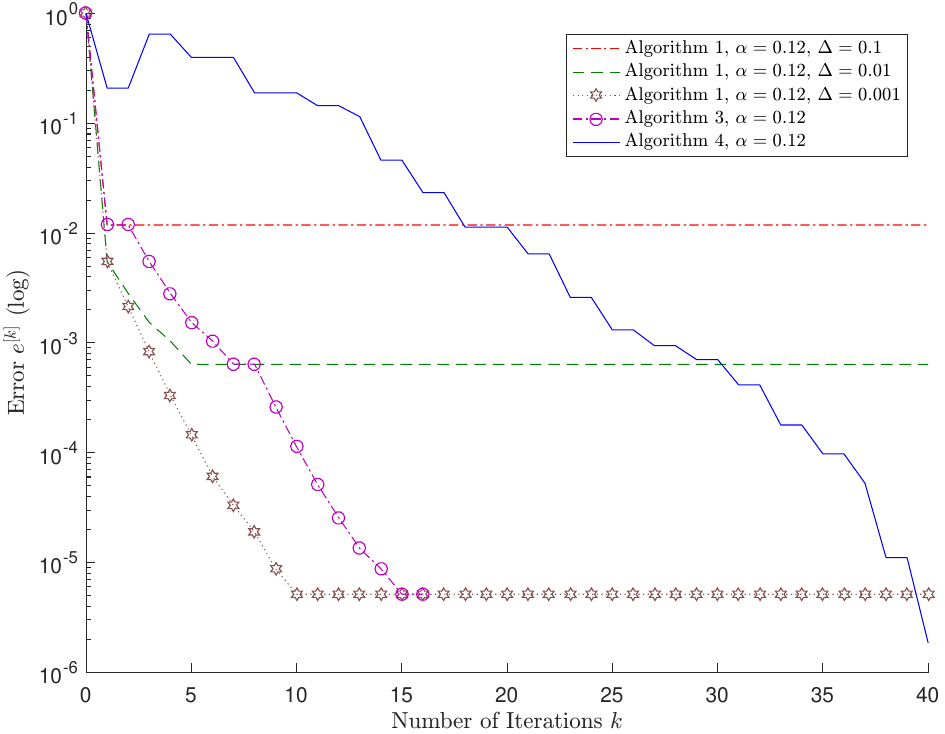}
    \caption{Execution of Algorithm~\ref{algorithm_1}, Algorithm~\ref{algorithm_2}, and Algorithm~\ref{algor_4} over a random digraph of $20$ nodes.}
    \label{20_nodes_V4}
\end{figure}

In Fig.~\ref{20_nodes_V4}, we focus on the  following two cases: 
\\ \noindent 
\textbf{A-A.} 
We evaluate the convergence speed and estimation precision of our algorithms. 
\\ \noindent 
\textbf{A-B.} 
We calculate the total communication (i.e., total transmitted bits) required in order to estimate the optimal solution for various levels of precision. 

\textbf{A-A. Convergence Speed and Estimation Precision.} 
In Fig.~\ref{20_nodes_V4} we plot Algorithm~\ref{algorithm_1} for different quantization levels. 
We can see that the optimal solution estimation precision depends on the quantization level. 
Additionally, we notice that Algorithm~\ref{algorithm_2} terminate their operation in time steps $16$, and $40$, respectively, due to performance criteria being fulfilled.  
In Fig.~\ref{20_nodes_V4}, we can see that Algorithm~\ref{algorithm_1} outperforms Algorithm~\ref{algorithm_2} which outperforms  Algorithm~\ref{algor_4} for all $\Delta$ values in terms of convergence speed. 
However, Algorithm~\ref{algorithm_2} and Algorithm~\ref{algor_4} enable nodes to achieve higher precision in approaching the optimal solution compared to Algorithm~\ref{algorithm_1}. 
This is due to their capability to refine the quantization level based on performance criteria. 
In contrast, Algorithm~\ref{algorithm_1} only enables nodes to converge to a specific neighborhood of the optimal solution (which is dependent on the utilized quantization level, without the ability to enhance precision further).
Additionally, in Algorithm~\ref{algorithm_1} and Algorithm~\ref{algorithm_2}, the number of bits per quantized message depends on the chosen quantization level. 
This implies that each message may consist of more than $3$ bits. 
Specifically, in Algorithm~\ref{algorithm_1} the messages may exceed $3$ bits for all $\Delta \in \{0.01, 0.001\}$. 
Furthermore, in Algorithm~\ref{algorithm_2} the progressive refinement of the quantization level results in an increasing number of bits in the messages exchanged among nodes. 
This increase leads to higher bandwidth requirements and may pose challenges related to bandwidth constraints within the network (as analyzed in \textbf{A-B.} below). 
In contrast, Algorithm~\ref{algor_4} can refine its quantization level and estimate the optimal solution while respecting the $3$-bit communication constraint among nodes.

\textbf{A-B. Total Communication.} 
In Table~\ref{bitcomparisons}, we compare the communication requirements (i.e., the number of bits communicated from each node per optimization time step) of Algorithm~\ref{algorithm_1}, Algorithm~\ref{algorithm_2}, and Algorithm~\ref{algor_4} for various levels of estimation precision as depicted in Fig.~\ref{20_nodes_V4}. 
The total number of transmitted bits $b_t$ is calculated as $b_t = c_s \cdot b_{pm} \cdot n_{tt}$, where $c_s$ is the convergence time step, $b_{pm}$ is the number of bits per transmitted message, and $n_{tt}$ is the total number of transmissions by nodes in the network. 
The convergence time step $c_s$ is shown in Fig.~\ref{20_nodes_V4}. 
The number of bits per transmitted message $b_{pm}$ depends on the quantization level. 
Specifically, for Algorithm~\ref{algor_4}, the transmitted messages consist of $3$ bits. 
For Algorithm~\ref{algorithm_2}, the transmitted messages consist of $7$ bits for steps $k = 0, 1, 2$, $10$ bits for steps $k = 3, \ldots, 8$, and $14$ bits for steps $k = 9, \ldots, 16$ due to the refinement of the quantization level. 
For Algorithm~\ref{algorithm_1}, the transmitted messages consist of $7$, $10$, and $14$ bits for $\Delta = 0.1, 0.01$, and $0.001$, respectively.
The total number of transmissions $n_{tt}$ for each node is determined by the operation of Algorithm~\ref{algorithm_1a}. 
It is important to note that our three optimization algorithms rely on Algorithm~\ref{algorithm_1a} for coordination among nodes. 
Algorithm~\ref{algorithm_1a} converges in a probabilistic manner. 
According to \cite[Table~$2$]{RIKOS2022110621}, the average total number of transmissions during the execution of Algorithm~\ref{algorithm_1a} is $211.88$.

\begin{center}
\captionof{table}{Total Number of Communicated Bits for Algorithm~\ref{algorithm_1}, Algorithm~\ref{algorithm_2}, and Algorithm~\ref{algor_4} over a random digraph of $20$ nodes.} 
\label{bitcomparisons} 
\begin{tabular}{|c||r|r|r|} 
\hline
Algorithm / $e^{[k]}$ & $10^{-2}$ & $10^{-3}$ & $10^{-5}$ \\
\cline{1-4} 
Alg.~\ref{algorithm_1} ($\Delta = 0.1$) & $4449.48$ & -- & -- \\ 
Alg.~\ref{algorithm_1} ($\Delta = 0.01$) & $6356.40$ & $10594.00$ & -- \\ 
Alg.~\ref{algorithm_1} ($\Delta = 0.001$) & $8898.96$ & $14831.60$ & $32629.52$ \\ 
Algorithm~\ref{algorithm_2} & $4449.48$ & $15043.28$ & $35807.72$ \\
Algorithm~\ref{algor_4} & $11441.52$ & $17162.28$ & $25425.60$ \\ 
\hline 
\end{tabular} 
\end{center}

In Table~\ref{bitcomparisons}, let us focus on the scenario where $e^{[k]}$ reaches $10^{-2}$ or $10^{-3}$. 
In this case, the total communication requirements of Algorithm~\ref{algor_4} exceed those of Algorithm~\ref{algorithm_1} and Algorithm~\ref{algorithm_2}. 
This is due to the faster convergence of Algorithm~\ref{algorithm_1} and Algorithm~\ref{algorithm_2}. 
Specifically, although Algorithm~\ref{algorithm_1} and Algorithm~\ref{algorithm_2} require nodes to exchange messages with more than $3$ bits, their rapid convergence results in lower overall communication requirements compared to Algorithm~\ref{algor_4}. 
However, it is noteworthy that when $e^{[k]}$ reaches $10^{-5}$, the communication requirements of Algorithm~\ref{algor_4} are significantly lower than those of Algorithm~\ref{algorithm_1} and Algorithm~\ref{algorithm_2}. 
To estimate the optimal solution with high precision, Algorithm~\ref{algorithm_1} and Algorithm~\ref{algorithm_2} require nodes to exchange messages consisting of $14$ bits. While this enables fast convergence, it also greatly increases the communication requirements. 
In contrast, Algorithm~\ref{algor_4} maintains low communication requirements because (i) the exchanged messages consist of $3$ bits, and (ii) the basis of the quantizer changes throughout the operation. 
Consequently, for achieving an error $e^{[k]}$ of $10^{-5}$ or lower, the communication requirements of Algorithm~\ref{algor_4} are significantly less than those of Algorithm~\ref{algorithm_1} and Algorithm~\ref{algorithm_2}. 

\subsection{Comparison with the Literature}\label{sec:Comparisons} 

We now compare Algorithm~\ref{algorithm_1}, Algorithm~\ref{algorithm_2}, and Algorithm~\ref{algor_4} against existing algorithms from the literature on static, strongly connected digraphs comprising of $20$ nodes. 
We display the error $e^{[k]}$ (as defined in \eqref{eq:distance_optimal}) on a logarithmic scale relative to the number of iterations. 
The error $e[k]$ was evaluated and averaged over $20$ trials. 
Note here that for the operation of the algorithms the parameters of each node $v_i$ are consistent with those detailed in Section~\ref{comparison_3algs} with the only difference being that we set $\varepsilon_{s_1} = 0$, $\varepsilon_{s_2} = 0$  for Algorithm~\ref{algorithm_2} and Algorithm~\ref{algor_4}. 

In Fig.~\ref{comp_v3} we can see that Algorithm~\ref{algorithm_2} is the fastest among our algorithms and other algorithms in the literature, outperformed only by \cite{2022:Jiang_Charalambous} (due to the fact that the latter enables nodes to exchange real valued messages). 
Additionally, Algorithm~\ref{algorithm_2} and Algorithm~\ref{algor_4} can either estimate the exact optimal solution (by setting $\varepsilon_{s_1} = 0$, $\varepsilon_{s_2} = 0$) or terminate their operation distributively based on performance guarantees--a characteristic not found in any other algorithms in the literature. 
This advantage is of particular importance since most algorithms assume that the messages exchanged among nodes are real numbers and admit asymptotic convergence within some error. 
Algorithm~\ref{algorithm_1} exhibits faster convergence compared to most algorithms until it reaches a neighborhood of the optimal solution (as it does not refine the quantization level). 
Lastly, while Algorithm~\ref{algor_4} is outperformed by our other two algorithms, its capability to refine the quantization level and adjust the quantizer's basis allows it to outperform most algorithms in the literature. 
Moreover, it is uniquely designed to function in environments with limited bandwidth using only $3$ bits per message. 
This characteristic is particularly advantageous in practical applications, although it affects the convergence rate as we can see in Fig.~\ref{comp_v3}. 

\begin{figure}[t]
    \centering
    \includegraphics[width=\linewidth]{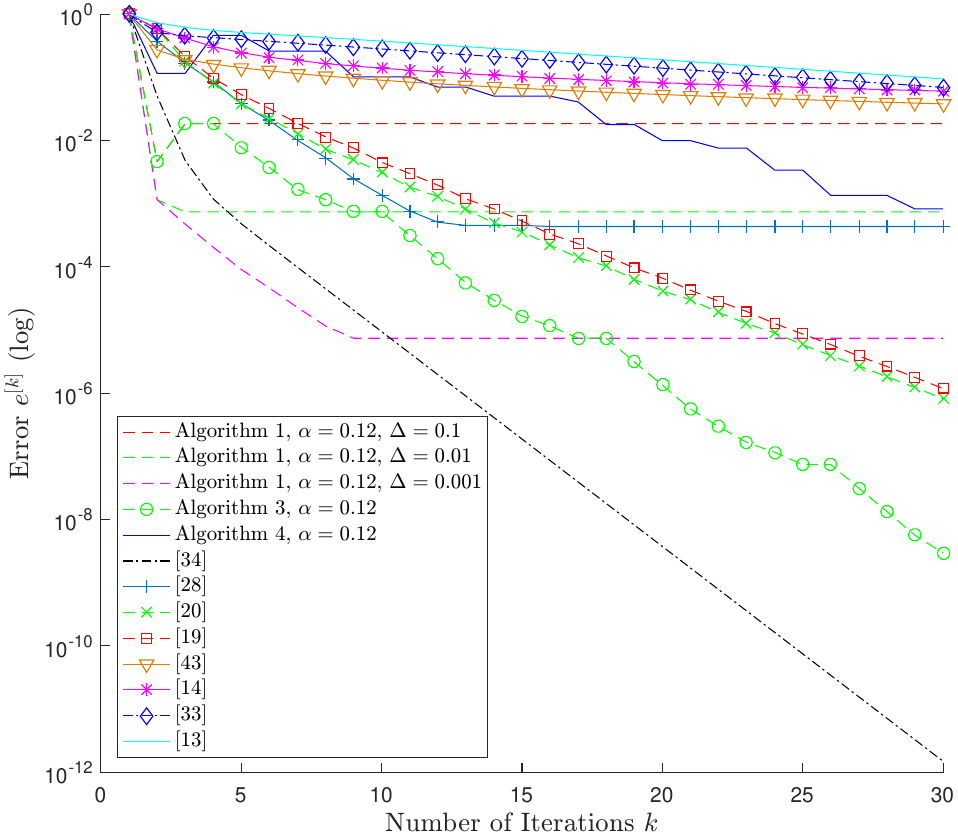}
    \caption{Error $e^{[k]}$ (defined in \eqref{eq:distance_optimal}) in logarithmic scale for Algorithm~\ref{algorithm_1}, Algorithm~\ref{algorithm_2}, Algorithm~\ref{algor_4}, with the approaches in \cite{2022:Jiang_Charalambous}, \cite{2024_Rikos_Themis_Johan_TCNS_CPU},
    \cite{2021:Nedic_PushPull}, 
    \cite{2018:Khan_AB},
    \cite{2020:Doostmohammadian_Charalambous},
    \cite{2018:Xie},
    \cite{2018:Khan_addopt}, and
    \cite{2009:Nedic_Optim}
    averaged over $20$ randomly generated strongly connected digraphs of $20$ nodes each. 
    }
    \label{comp_v3}
\end{figure}

\section{Conclusions and Future Directions}\label{sec:conclusions}

In our paper we focused on the unconstrained distributed optimization problem. 
We introduced three distributed optimization algorithms that exhibit efficient communication among nodes by enabling them to exchange quantized messages. 
We analyzed their operation and established their linear convergence. 
Additionally, we demonstrated their capability to either enhance the estimation of the optimal solution or terminate their operation according to predefined performance guarantees in a distributed manner. 
Finally, we presented a motivating application on distributed target localization, and concluded our paper with comparisons against other algorithms from the literature. 

Our main future research direction encompasses the extension of our algorithms to handle constrained optimization problems that are prevalent in many real-world applications. 
Additionally, we plan to investigate the performance of our algorithms in open networks (i.e., networks where the number of nodes and edges changes over time). 

\section*{Acknowledgments} 
The authors would like to gratefully acknowledge Professor Nicolò Michelusi (Arizona State University) in formulating and refining the proof, which significantly strengthened the validity and rigor of our findings. 

\appendices
%
%
%
%
\section{Proof of Theorem~\ref{theorem_convergence_stronglyConvex}}
\label{appen_theorem}

Before presenting the proof of our theorem, we first present the following lemma. 
\begin{lemma}
\label{lemma_rule}
	For any function $f(x)$ which is $\mu $ strongly convex and $L$ smooth, let $\beta \le \frac{2}{\mu + L}$, then, $\forall x_1, x_2 \in \mathbb{R}^p$, the following inequality holds:
	\begin{align}\label{rule}
\|x_1-x_2  - \beta ( \nabla f(x_1) - \nabla f(x_2) ) \|  
\le (1-\mu \beta )\|x_1-x_2\|.
\end{align}
\end{lemma}
\begin{proof} Similarly to the proof of Theorem 1 in~\cite{10279097}, we have that
\begin{align}
\|x_1-x_2  -& \beta ( \nabla f(x_1) - \nabla f(x_2) ) \| \\
=& ( \|x_1-x_2\|^2  \beta^2 \| \nabla f(x_1) - \nabla f(x_2)\|^2  \nonumber  \\
  &- 2\beta (x_1-x_2)^\top (\nabla f(x_1) - \nabla f(x_2)) )^{\frac{1}{2}}\nonumber\\
 \le& ( (1- \beta\frac{2\mu L}{\mu + L})  \|x_1-x_2\|^2 \nonumber \\ &-\beta (\frac{2}{\mu + L} -\beta )\| \nabla f(x_1) - \nabla f(x_2)\|^2 ) ^{\frac{1}{2}} \label{rule_1}  \\
\le& (1-\mu \beta )\|x_1-x_2\|, \label{rule_2}
\end{align} 
where \eqref{rule_1} comes from \cite[Theorem 2.1.12]{nesterov2003introductory}, (i.e., $(x_1-x_2)^\top (\nabla f(x_1) - \nabla f(x_2)) \ge \frac{\mu L}{\mu +L} \|x_1-x_2\|^2 + \frac{1}{\mu +L}\| \nabla f(x_1) - \nabla f(x_2)\|^2$), and \eqref{rule_2} comes from $ \| \nabla f(x_1) - \nabla f(x_2)\|^2 \ge \mu \|x_1-x_2\|^2 $ and $\beta \le \frac{2}{\mu + L}$.
\end{proof}

\begin{remark}
The contraction property of gradient descent can be also proven using the Mean Hessian Theorem \cite[\S~II.E]{larsson2025unified}.
\end{remark}

We are now ready to show the proof of Theorem~\ref{theorem_convergence_stronglyConvex}. 
From~\eqref{hatx_calcu2} we have 
\begin{align}
\hat{x}^{[k+1]} -x^{*} 
=&(\hat{x}^{[k]}-x^{*} - \hat{\alpha}  \hat u^{[k]}) +( \hat{\alpha}  \hat u^{[k]}- \hat \alpha u^{[k]}) + e^{[k+1]}. 
\end{align}
Combining $\sum_{i=1}^{n} \nabla f_i(x^{*})=\mathbf{0}$ with  $\hat{\alpha} = \frac{\alpha}{n}$, and \eqref{rule} in Lemma~\ref{lemma_rule}, when 
\begin{align}\label{alpha_upperbound}
\alpha \le \frac{2n}{\mu +L}, \ L =  \sum_i L_i, \ \text{and} \ \mu = \min_i \mu_i, 
\end{align} 
we have
\begin{align}
\| \hat{x}^{[k]}-&x^{*} - \hat{\alpha} \hat u^{[k]}\| \\
=&\|\hat{x}^{[k]}-x^{*} - \hat{\alpha} (\sum_{i=1}^{n} \nabla f_i(x_i^{[k]}) - \sum_{i=1}^{n} \nabla f_i(x^{*})) \|\nonumber \\
\le& (1-\hat{\alpha}\mu )\|\hat{x}^{[k]}-x^{*}\|\label{first_bound}. 
\end{align}
Let us denote $ \tilde{x}^{[k]} = \hat{x}^{[k]} - x^{*} $.
Combining \eqref{first_bound} with \eqref{uk_hat} and $\|e^{[k]}\|\le 2\Delta$, we arrive at 
\begin{align}\label{before_final_bound}
\| \tilde{x}^{[k+1]}\|
&\le (1-\hat{\alpha}\mu )\| \tilde{x}^{[k]}\| +  4\hat{\alpha}L\Delta + 2\Delta.
\end{align}

\noindent
Finally, let us denote $\vartheta \coloneqq 1-\frac{\alpha \mu}{n}$, and $\mathcal{O}(\Delta) 
		 =( \frac{4\alpha L}{n} + 2)\Delta$.  
From~\eqref{before_final_bound} we get 
\begin{equation}
\begin{aligned}
\| \hat{x}^{[k+1]} - x^{*}\|
\le&
\vartheta\| \hat{x}^{[k]} - x^{*}  \|   + \mathcal{O}(\Delta)\\
\le&  \vartheta^2\| \hat{x}^{[k-1]} - x^{*}  \| + (\vartheta+1)\mathcal{O}(\Delta)\\
\le&  \vartheta^3\| \hat{x}^{[k-2]} - x^{*}  \| + (\vartheta^2+ \vartheta+1)\mathcal{O}(\Delta)
\\
\le& \vartheta^{k+1}\| \hat{x}^{[0]} - x^{*}  \| + \frac{1-\vartheta^{k+1}}{1-\vartheta}\mathcal{O}(\Delta).\label{final_proof_bound}
\end{aligned}
\end{equation}
This concludes the proof of our theorem. 
Let us finally note that considering \eqref{alpha_upperbound} and for $\alpha \in (0, \frac{2n}{\mu + L}]$ we have
\begin{align}\label{alpha_interval}
1>1-\hat{\alpha}\mu \ge 1-\frac{2\mu}{\mu +L} = \frac{L-\mu}{\mu +L} \ge 0. 
\end{align}  
Thus, in~\eqref{final_proof_bound} when $k \rightarrow \infty $, based on~\eqref{alpha_interval}, we have 
\begin{equation}
\begin{aligned}
\| \hat{x}^{[k]} - x^{*}\|
< \frac{n}{\alpha \mu}\mathcal{O}(\Delta).
\end{aligned}
\end{equation}
These results signify that during the operation of Algorithm~\ref{algorithm_1} the state of each node linearly converges to a neighborhood of the optimal solution $x^*$ (as shown in Section~\ref{sec:results}).

%
%
%
%



\bibliographystyle{IEEEtran}
\bibliography{bibliography}

\begin{IEEEbiography}[{\includegraphics[width=1in,height=1.25in,clip,keepaspectratio]{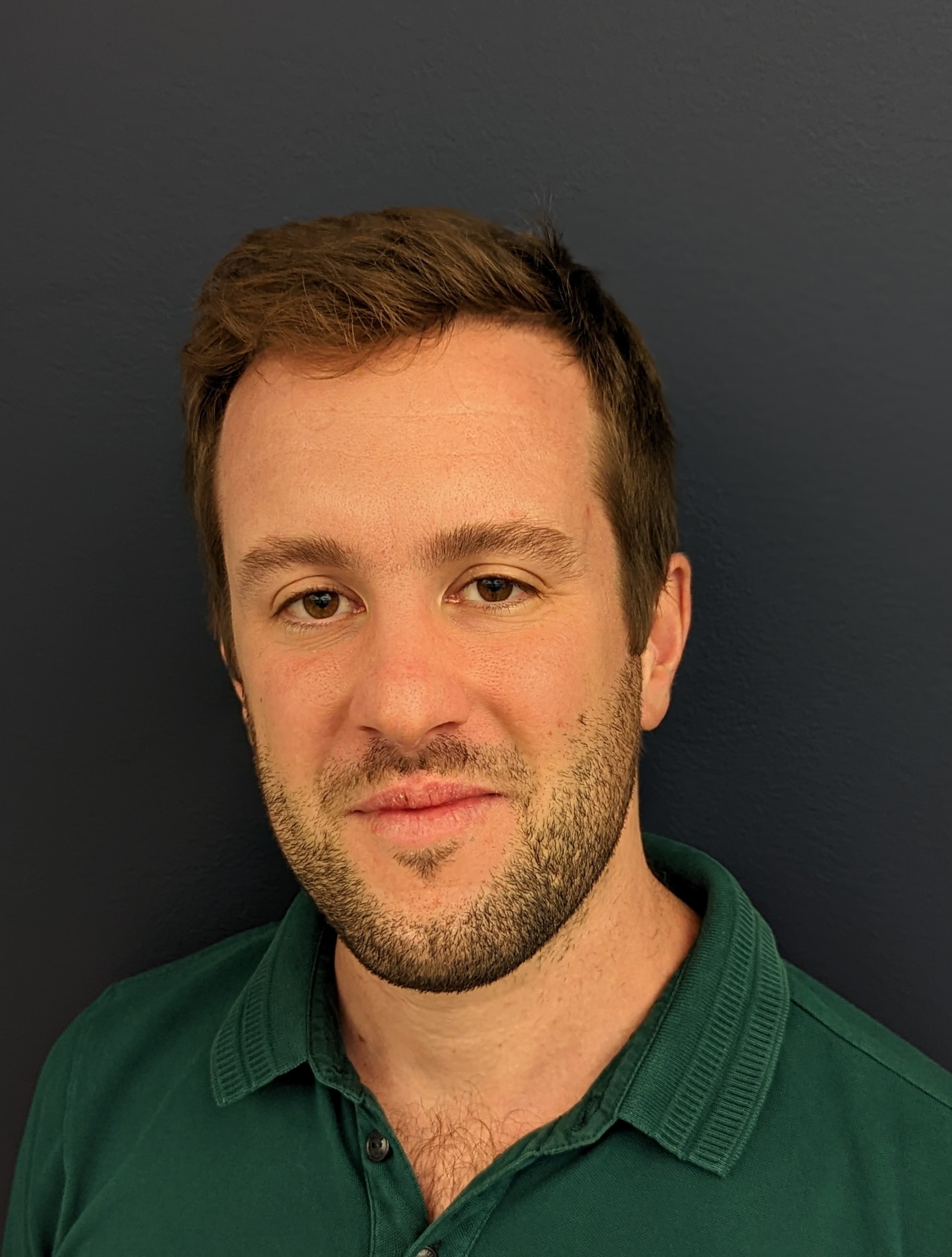}}]{Apostolos I. Rikos} (M'16) is an Assistant Professor at the Artificial Intelligence Thrust of the Information Hub, The Hong Kong University of Science and Technology (Guangzhou), Guangzhou, China. 
He is also affiliated with the Department of Computer Science and Engineering, The Hong Kong University of Science and Technology, Clear Water Bay, Hong Kong, China.
He received his B.Sc., M.Sc., and Ph.D. degrees in Electrical Engineering from the Department of Electrical and Computer Engineering, University of Cyprus in 2010, 2012, and 2018, respectively.
In 2018, he joined the KIOS Research and Innovation Center of Excellence in Cyprus, where he was a Research Lecturer. 
In 2020, he joined the Division of Decision and Control Systems of KTH Royal Institute of Technology as a Postdoctoral Researcher. 
In 2023, he joined the Department of Electrical and Computer Engineering, Division of Systems Engineering, at Boston University as a Postdoctoral Associate.
His research interests are in the area of distributed optimization and learning, distributed network control and coordination, privacy and security. 
\end{IEEEbiography}

\vspace{-0.2cm}

\begin{IEEEbiography}[{\includegraphics[width=1in,height=1.25in,clip,keepaspectratio]{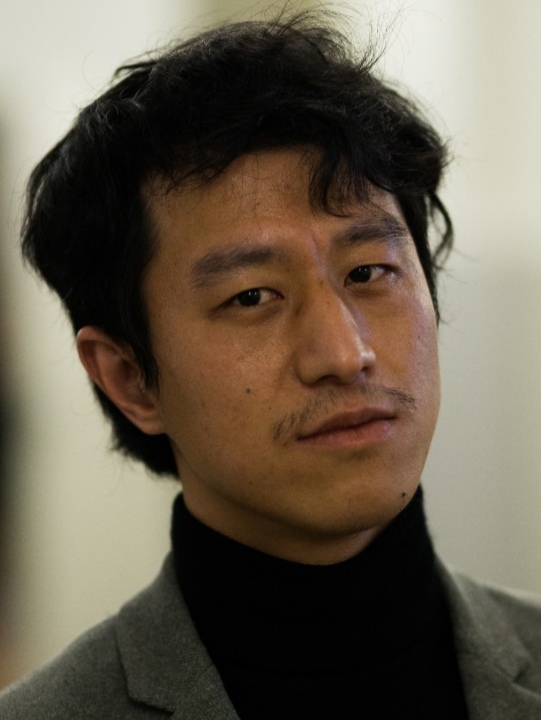}}]{Wei Jiang} 
received his B.S. degree in mechanical engineering and automation from Wuhan University of Technology, Wuhan, China, in 2011, and M.S. degree in automobile engineering from Beihang University, Beijing, China, in 2015 and Ph.D. degree in Automatic, Computer Engineering, Signal Processing and Images in CRIStAL, UMR CNRS 9189, Ecole Centrale de Lille, France, in 2018.
He was a postdoctoral researcher at Aalto University, Finland during 2019-2022. 
He was a Visiting Scholar at University of Cyprus, Cyprus and KU Leuven, Belgium. 
His research interests include distributed optimization/control, learning algorithms, vehicle platooning, and robotics.
\end{IEEEbiography}

\vspace{-0.2cm}

\vspace{-0.2cm}

\begin{IEEEbiography}[{\includegraphics[width=1in,height=1.25in,clip,keepaspectratio]{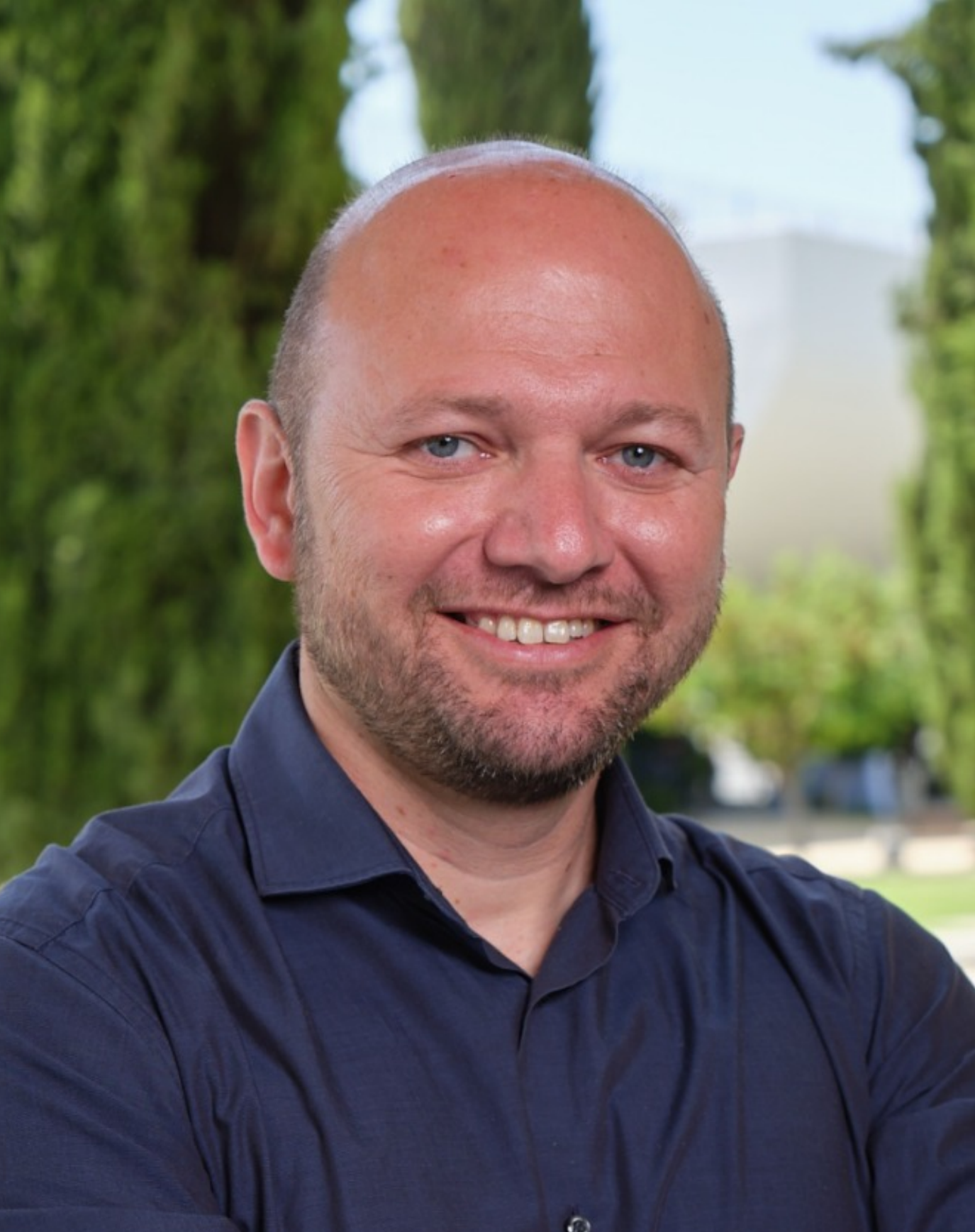}}]{Themistoklis Charalambous} (S'05, M'10, SM'20)
received his BA and M.Eng in Electrical and Information Sciences from Trinity College, Cambridge University. 
He completed his PhD studies in the Control Laboratory of the Engineering Department, Cambridge University. 
Following his PhD, he held postdoctoral positions at Imperial College London, at the Royal Institute of Technology (KTH), and at Chalmers University of Technology. 
In January 2017, he joined Department of Electrical Engineering and Automation, School of Electrical Engineering, Aalto University as a tenure-track Assistant Professor. In September 2018, he was awarded the Academy of Finland Research Fellowship and in July 2020 he was appointed as a tenured Associate Professor. In September 2021, he joined the Department of Electrical and Computer Engineering, University of Cyprus as a tenure-track Assistant Professor and he remains associated with Aalto University as a Visiting Professor. Since April 2023, he is also a Visiting Professor at the FinEst Centre for Smart Cities.

His primary research targets the design and analysis of (wireless) networked control systems that are stable, scalable and energy efficient. The study of such systems involves the interaction between dynamical systems, their communication and the integration of these concepts. As a result, his research is interdisciplinary combining theory and applications from control theory, communications, network and distributed optimization. 
\end{IEEEbiography}

\vspace{-0.2cm}

\begin{IEEEbiography}[{\includegraphics[width=1in,height=1.25in,clip,keepaspectratio]{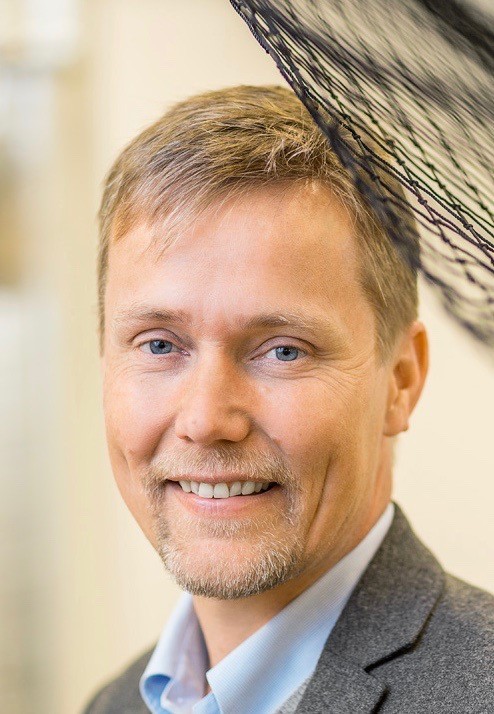}}]{Karl H. Johansson}
is Swedish Research Council Distinguished Professor in Electrical Engineering and Computer Science at KTH Royal Institute of Technology in Sweden and Founding Director of Digital Futures. He earned his MSc degree in Electrical Engineering and PhD in Automatic Control from Lund University. He has held visiting positions at UC Berkeley, Caltech, NTU and other prestigious institutions. His research interests focus on networked control systems and cyber-physical systems with applications in transportation, energy, and automation networks. For his scientific contributions, he has received numerous best paper awards and various distinctions from IEEE, IFAC, and other organizations. He has been awarded Distinguished Professor by the Swedish Research Council, Wallenberg Scholar by the Knut and Alice Wallenberg Foundation, Future Research Leader by the Swedish Foundation for Strategic Research. He has also received the triennial IFAC Young Author Prize and IEEE CSS Distinguished Lecturer. He is the recipient of the 2024 IEEE CSS Hendrik W. Bode Lecture Prize. His extensive service to the academic community includes being President of the European Control Association, IEEE CSS Vice President Diversity, Outreach $\&$ Development, and Member of IEEE CSS Board of Governors and IFAC Council. He has served on the editorial boards of Automatica, IEEE TAC, IEEE TCNS and many other journals. He has also been a member of the Swedish Scientific Council for Natural Sciences and Engineering Sciences. He is Fellow of both the IEEE and the Royal Swedish Academy of Engineering Sciences.
\end{IEEEbiography}

%
%
%
%
\end{document}